\documentclass{article}
\usepackage[utf8]{inputenc}

\usepackage[english]{babel}
\usepackage{color}
\usepackage{graphicx}

\usepackage{amsthm,amsmath, amsfonts}
\usepackage{amssymb,verbatim}

\usepackage{algorithm,algpseudocode}

\usepackage{color}
\usepackage[dvipsnames]{xcolor}

\definecolor{mygray}{gray}{0.1}
\usepackage[a4paper, margin=2cm]{geometry}

\usepackage[toc,page]{appendix}
\usepackage{authblk}

\definecolor{DarkMidnightBlue}{rgb}{0.0, 0.2, 0.5}
\usepackage[colorlinks=true,linkcolor=Maroon,citecolor=DarkMidnightBlue,urlcolor=black]{hyperref}

\usepackage{lineno}

\usepackage{pgfplots}
\usetikzlibrary{patterns}
\usetikzlibrary{math}
\usepgfplotslibrary{fillbetween}

\makeatletter
\def\lstAZ{A, B, C, D, E, F, G, H, I, J, K, L, M, N, O, P, Q, R, S, T, U, V, W, X, Y, Z}
\def\lstaz{a, b, c, d, e, f, g, h, i, j, k, l, m, n, o, p, q, r, s, t, u, v, w, x, y, z}

\def\lstAZBB{B, C, D, E, F, G, H, I, J, K, L, M, N, O, P, Q, R, T, U, V, W, X, Y, Z}

\newcommand{\MkScr}[1]{\expandafter\def\csname s#1\endcsname{\mathscr{#1}}}
\newcommand{\MkUp}[1]{\expandafter\def\csname u#1\endcsname{\mathrm{#1}}}
\newcommand{\MkFrak}[1]{\expandafter\def\csname f#1\endcsname{\mathfrak{#1}}}
\newcommand{\MkCal}[1]{\expandafter\def\csname c#1\endcsname{\mathcal{#1}}}
\newcommand{\MkBB}[1]{\expandafter\def\csname #1#1\endcsname{\mathbb{#1}}}

\@for\i:=\lstAZ\do{%
	\expandafter\MkScr \i  %
	\expandafter\MkFrak \i  %
	\expandafter\MkUp \i %
	\expandafter\MkCal \i  %
		  }    
\@for\i:=\lstaz\do{%
	\expandafter\MkUp \i   }    
	
\@for\i:=\lstAZBB\do{%
	\expandafter\MkBB \i     }
	    
\makeatother

\newcommand{\rr}{\mathrm{r}}
\newcommand{\bb}{\mathrm{b}}

\newcommand{\Vol}{\mathrm{Vol}}

\renewcommand{\P}{\mathrm{P}}

\newcommand{\ind}{\mathbf{1}}

\newcommand{\dd}{ \mathrm{d}}

\DeclareMathOperator*{\argmin}{argmin}

\renewcommand{\epsilon}{\varepsilon}

\renewcommand{\tilde}{\widetilde}

\numberwithin{equation}{section}

\newtheorem{theorem}{Theorem}[section]
\newtheorem{lemma}[theorem]{Lemma}
\newtheorem{proposition}[theorem]{Proposition}
\newtheorem{corollary}[theorem]{Corollary}

\newtheorem{definition}[theorem]{Definition}
\newtheorem{remark}[theorem]{Remark}

\newtheorem{assumption}[theorem]{Assumption}
\newtheorem{example}[theorem]{Example}

\usepackage{stmaryrd}

\newcommand{\R}{\mathbb R}

\definecolor{lime}{HTML}{A6CE39}
\DeclareRobustCommand{\orcidicon}{%
	\begin{tikzpicture}
	\draw[lime, fill=lime] (0,0) 
	circle [radius=0.16] 
	node[white] {{\fontfamily{qag}\selectfont \tiny ID}};
	\draw[white, fill=white] (-0.0625,0.095) 
	circle [radius=0.007];
	\end{tikzpicture}
	\hspace{-2mm}
}

\foreach \x in {A, ..., Z}{%
	\expandafter\xdef\csname orcid\x\endcsname{\noexpand\href{https://orcid.org/\csname orcidauthor\x\endcsname}{\noexpand\orcidicon}}
}

\title{Methods and applications of PDMP samplers\\with boundary conditions}

\author[1]{Joris Bierkens}
\author[2]{Sebastiano Grazzi}
\author[2]{Gareth Roberts}
\author[3]{Moritz Schauer}
\affil[1]{Delft Institute of Applied Mathematics (DIAM), TU Delft, The Netherlands}
\affil[2]{Department of Statistics, University of Warwick, UK}
\affil[3]{Department of Mathematical Sciences Chalmers University of Technology and University of Gothenburg, Sweden}
\usepackage[backend=bibtex, citestyle=authoryear]{biblatex}  
\date{\today}
\pgfplotsset{compat=1.17}
\begin{document}
\maketitle
\begin{abstract}
    \noindent We extend Monte Carlo samplers based on piecewise deterministic Markov processes (PDMP samplers) by formally defining different boundary conditions such as \emph{sticky floors}, \emph{soft} and \emph{hard walls} and \emph{teleportation portals}. This allows PDMP samplers to target measures with piecewise-smooth densities relative to mixtures of Dirac and continuous components and measures supported on disconnected regions or regions which are difficult to reach with continuous paths. This is achieved by specifying the transition kernel which governs the behaviour of standard PDMPs when reaching a boundary. We determine a sufficient condition for the kernel at the boundary in terms of the \emph{skew-detailed balance condition} and give concrete examples. The probabilities to cross a boundary can be tuned by introducing a piecewise constant speed-up function which modifies the velocity of the process upon crossing the boundary without  extra computational cost. We apply this new class of processes to two illustrative applications in epidemiology and statistical mechanics. 
\end{abstract}

\section{Overview}\label{sec: introduction}
\subsection{Introduction}
Markov Chain Monte Carlo  (MCMC) methods are vital tools used to explore distributions in Bayesian inference and other applications. They involve simulating Markov chains which converge in distribution to the desired target measure. The efficiency of MCMC is determined by the convergence properties of the resulting Markov chain and the computational cost involved in simulating it. For these reasons, recent attention has been drawn to Monte Carlo methods based on continuous-time piecewise deterministic Markov processes (PDMP samplers),  see \textcite{vanetti2017piecewisedeterministic} for an overview. PDMP samplers are non-reversible continuous-time Markov processes endowed with \emph{momentum} and characterized by deterministic dynamics interrupted by  a finite collection of random events at which the process jumps in location and/or momentum. This class of processes have been shown to have good mixing properties (fast convergence to the target measure, see for example \cite{diaconis2000analysis}), have low asymptotic variance (see for example \cite{CHEN20131956}), and can be simulated exactly in continuous time (up to floating point precision). Another attractive feature of PDMP samplers is that they allow the substitution of the gradient of the target log-density by an unbiased estimate thereof without introducing bias. This technique is refereed as \emph{exact subsampling} and leads to efficient simulations in difficult scenarios e.g. it has been exploited in regression problems with large sample size (\cite{bierkens2019}, \cite{bierkens2020boomerang}, \cite{bierkens2021sticky}) or for high-dimensional problems with intractable densities (\cite{bierkens2020piecewise}). 

Current work has almost exclusively concentrated on smooth differential target densities with respect to a fixed dimensional Lebesgue measure and in infinite domain without the presence of natural boundaries. In that context, a rather complete theory can be established, see \textcite{vanetti2017piecewisedeterministic}. However many interesting applications lie outside that area, for example when using Bayesian Lasso methods for variable selection (\cite{park2008bayesian}) and for hard-sphere stochastic geometry models. For the latter application, PDMP samplers have gained popularity and have been extensively studied in the field of statistical physics under the name of \emph{event-chain Monte Carlo} methods, see for example \textcite{PhysRev},  \textcite{michel2014generalized}, \textcite{10.3389/fphy.2021.663457}.
There is therefore a compelling need to provide a general framework for PDMPs outside the smooth density setting. 
This article will address this need, providing both a theoretical framework and a general and practical methodology for PDMPs in the non-smooth case.

The contributions of this work are:
\begin{itemize}
    \item We provide a general framework for the construction of PDMPs targeting a density with discontinuities along a boundary. We show how the behaviour of the target density at the boundary of its support gives rise to a simple condition at the boundary in  terms of skew-detailed balance which ensure the invariance of the target density. 
    \item We provide rigorous theory to underpin our framework which includes results on invariance and ergodicity for our generalised PDMP processes.
    \item The introduced framework enables us to precisely define complex behaviors of the process at the boundary, giving rise to a variety of distinct features informally referred here as \emph{soft} and \emph{hard} walls, \emph{teleportation portals}, and \emph{sticky} floors.
 \item We provide an analogous framework for PDMPs to target densities supported on disconnected sets, and those
which have a density relative to mixtures of continuous and atomic components. 
    \item We provide a practical and efficient methodology for implementing these generalised PDMP samplers and provide simple illustrations, including  sampling the invariant measure in hard-sphere stochastic geometry models. 
    \item Our methodology allows us to apply PDMP samplers with boundary conditions to a completely novel  example: for sampling the latent space of infected times with unknown infected population size in the SIR epidemic model with notifications.

\end{itemize}
Our approach generalizes \textcite{Bierkens_2018} which defines the Bouncy Particle Sampler (a popular PDMP sampler) for target densities on restricted domains. In this work, we allow PDMPs to change speed (magnitude of their velocity) according to a \emph{speed-up} function.
The use of speed-up functions was initially considered in  \textcite{vasdekis2021speed} for sampling efficiently heavy tailed distributions with Zig-Zag samplers. \textcite{vasdekis2021speed} only considered the class of continuous speed-up functions for which the deterministic dynamics could be computed analytically. In this work, we extend this class by considering functions which are discontinuous at the boundary. We combine our framework with the sticky events presented in \cite{bierkens2021sticky} for targeting measures which are also mixture of continuous and atomic component. The framework considered allows to define boundaries which acts as teleportation portals, allowing the process to jump in space along the given boundaries. The teleportation portals defined in this work are reminiscent of the approach of \textcite{moriarty2020metropolis} for sampling measures in a non-convex and disconnected space with a Metropolis-Hastings algorithm.
Another approach based on the Hamiltonian Monte Carlo (HMC) sampler for sampling discontinuous densities is given by \textcite{Nishimura_2020}. In contrast with HMC, the PDMPs presented here can be simulated in continuous time, without relying on discretization methods.

Our work complements the theoretical study of \textcite{chevallier2021pdmp}, which derives rigorously the invariant measure for PDMPs with boundary conditions.
Contrary to their work, our invariance theory takes advantage of a natural skew-detailed balance condition on the boundary, and this enables us to consider more general classes of PDMPs, in particular incorporating jumps in location (introducing teleportation portals), and allowing us to change speed by means of a speed-up function which can be discontinuous at the boundary. Similarly to \textcite{chevallier2021pdmp}, we propose a Metropolised step of PDMPs at the boundary and give concrete examples such as the Bouncy particle samplers and the Zig-Zag sampler. In our work, we also discuss ergodicity, we extend the Lebesgue reference measure to mixtures of Dirac and continuous components, and include two challenging applications of independent interest.

\subsection{An illustrative example}\label{sec: toy}
PDMP samplers can be used for targeting a wide class of multi-dimensional measures. In this section, we give an artificial example which informally illustrates the rich behaviour of the class of processes considered in this article.

As an example of a PDMP sampler we consider the Zig-Zag sampler with boundary conditions, featuring a rich behaviour given by random events of different nature. Figure~\ref{fig: teaser} displays the first two coordinates of a simulated trajectory. Below, we informally distinguish and comment on each of those random events and link them to their specific role in targeting an (artificially chosen) measure on $\RR^d$, with piecewise-smooth density proportional to $\exp(-\Psi(x))$ with
 \begin{equation}\label{eq: density showcase}
 \Psi(x) := x'\Gamma x - \sum_{i=1}^d \ind_{(x_i > 1/2)} c - \sum_{i=1}^{\lfloor d/2\rfloor}\log(x_{2i-1}) 
 \end{equation}
 relative to a reference measure
 \begin{equation}\label{eq: ref showcase}
 \prod_{i=1}^d \left(\ind_{(x_i \in [0,1])} \dd x_i + \delta_{0}(\dd x_i - \frac{1}{4})\right)
 \end{equation}
for a parameter $c>0$ and a matrix $\Gamma = 1.3 I + C 0.5$, where each element $C_{i,j}, \,i,j =1,2,\dots,d$ is 0 with probability 0.9 and an independent realization from $\cN(0,1)$ otherwise. For this simulation, we fixed the dimensionality $d = 80$.

\begin{itemize}
    \item \emph{(Random events and repelling walls)} Analogously to the standard Zig-Zag samplers, the process has constant velocity on the finite velocity  space $\{-1,+1\}^d$ and changes direction at random times by switching every time the sign of only one velocity component. This produces changes in direction and allows the process to target the smooth components of $\Psi(x)$ (the first and the third term of~\eqref{eq: density showcase}). In this example, the random events are such that  the position of the odd coordinates $\{(x_{i2 -1},v_{i2-1})\}_{i=1,2,\dots}$ can be arbitrarily close to 0, yet without ever touching $0$, giving rise to repelling walls. This is because the density vanishes on those hyper-planes. 
    \item \emph{(Sticky floors)} All coordinates $\{(x_i, v_i)\}_{i=1,2,\dots}$, upon hitting $\frac{1}{4}$, ``stick'' in that point for an exponentially distributed time. This corresponds to momentarily setting  the $i$th velocity component to $0$ and allows the process to spend positive time in hyper-planes $S_A = \bigotimes_{i=1}^d E_i$ where 
    \[
    E_i = \begin{cases}
        \{\frac{1}{4}\} & i \in A,\\
        \RR & \text{ otherwise.}
    \end{cases}
    \]
    for all $A \subset \{1,2,\dots,d\}$ (with some coefficients exactly equal to $1/4$). Sticky floors allow the process to target mixtures of continuous and atomic components and, in this example, allow to change the reference measure from a $d$-dimensional Lebesgue measure to~\eqref{eq: ref showcase}.
    \item \emph{(Soft walls)} The second term in~\eqref{eq: density showcase}) is such that  target density is discontinuous at $\frac{1}{2}$ in each coordinate. PDMPs targets these discontinuities by letting each coordinate process $(x_i,v_i),\, i = 1,2,\dots,d$ switch the sign of its velocity component  with some probability, upon hitting $\frac{1}{2}$ from below, while always crossing $\frac{1}{2}$ when hitting this point from above.
    \item \emph{(Hard walls)} The process switches always velocity at the boundaries $\{(x_{i2},v_{i2}) = (0,-1)\}_{i=1,2,\dots}$ and $\{((x_{i}, v_{i}) = (1, +1)\}_{i=1,2,\dots}$. This allows the process to explore only the regions in $\RR^d$ supported by the measure. 
\end{itemize}
Each bullet point in this list will be formalized and described in details in the subsequent sections. 

This is a constructed example which is of interest for multiple reasons:
\emph{i)} an efficient and local implementation of the Zig-Zag sampler can be adopted which greatly profits of the local dependence structure of $\mu$ implied by the sparse form of $\Gamma$ in \eqref{eq: density showcase} (see \cite[Section 4]{bierkens2020piecewise}); \emph{ii)} the continuous and atomic components of the reference measure (equation~\eqref{eq: ref showcase}) makes the sampling problem not trivial. A mixture of atomic and continuous components arises naturally for example in Bayesian variable selection with spike-and-slab priors. By including sticky events, PDMPs can efficiently sample from such mixture measures,  see \textcite{bierkens2021sticky} for more details;
\emph{iii)} As  $\lim_{y\downarrow 0} \Psi(x[2i \colon y]) = \infty,$\footnote{Here $y = x[i 
\colon c]$, for $c \in \RR$ means $y_j = c$ if $j = i$, $y_j = x_j$ otherwise.}  for $i=1,2,\dots$, the gradient of the log-likelihood explodes thus complicating the application of gradient-based Markov chain Monte Carlo methods; \emph{iv)} the discontinuities at $1/2$ and boundaries at $0$ and $1$ deteriorate the performance of ordinary MCMC (see \cite{neal2011mcmc}) and complicates the application of gradient-based methods, as the gradient is not defined at discontinuity. 
\begin{figure}[ht!]
    \centering
    \includegraphics[width=1\linewidth]{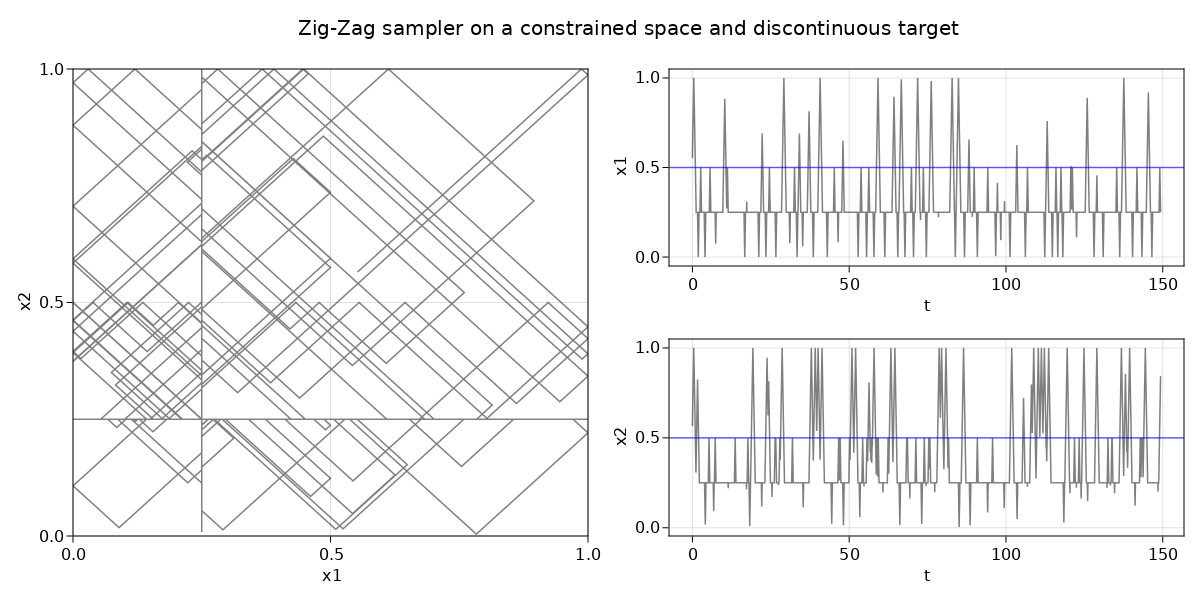}
    \caption{ $(x_1$-$x_2)$ phase space plot (left) and trace plots (right) of the first 2 coordinates of a Zig-Zag trajectory sampling a general density $f$ supported in $[0,1]^d$ with  discontinuity at $1/2$ (blue lines) in each coordinate with density vanishing at $x_2 = 0$. The reference measure has a Dirac mass at 1/4 in each coordinate.}
    \label{fig: teaser}
\end{figure}

\subsection{Outline}
Section~\ref{sec: pdmps samplers with boundaries}  presents the theoretical framework of PDMPs with boundary conditions. 
The invariant measure of PDMP samplers is established by reviewing classical conditions of the process made on the interior (Section~\ref{subsec: conditions on the interior}) and establishing new conditions at the boundary (Section~\ref{subsec: conditions at the boundary}). In Section~\ref{sec: ergodicity pdmp with boundaries} we discuss the limiting behaviour of the process and in Section~\ref{sec: sticky extensions} we extend the framework presented for considering dominating measures which are mixtures of atomic and continuous components.
We describe both the Zig-Zag samplers and the Bouncy particle samplers with boundary conditions in Section~\ref{sec: Zig-Zag sampler for discontinuous densities}-\ref{sec:  Bouncy Particle sampler with teleportation} and apply the former algorithm for the main application in epidemiology (Section~\ref{sec: epidemiology}) and the latter algorithm for the application in statistical mechanics (Section~\ref{sec: teleportation}). In Section~\ref{sec: discussion} we discuss promising research directions that arise from this work. 

\section{PDMP samplers with boundaries}\label{sec: pdmps samplers with boundaries}
In this section, we present an overview of PDMP samplers with boundary conditions, we review the known sufficient conditions which allow the process to target the smooth components of the density and establish new general conditions at the boundary based on \emph{skew-detailed balance} to target the discontinuity of the density at the boundary and to define teleportation portals. Remark~\ref{rmk: equvalence with weighted PDMP} highlights the connection between the speed-up function and importance weights.

We consider the problem of sampling from a measure $\mu$ defined on $\Omega = \bigsqcup_{i\in K} \Omega_i$ where $\Omega_i\subset \RR^d, \, i \in K$ are disjoint open sets and $K$ is a countable set. The measure $\mu$ is assumed to be of the form
\begin{equation}
    \label{eq: general form target measure}
    \mu(\dd x) := C\exp(-\Psi(x))\prod_{i=1}^d \mu_i(\dd x_i)
\end{equation} 
where 
\begin{equation} \label{eq: dominated measure}
\mu_i(\dd x_i) := \begin{cases} \dd x_i & \text{ if } i\in A,\\
\dd x_i + \frac{1}{\kappa_i(x)}\delta_{c_i}(\dd x_i) & \text{ otherwise,}
    \end{cases}
\end{equation}
for a set $A \subset \{1,2,\dots,d\}$, some positive constants $C>0,\, c_i>0$,  functions $\kappa_i \colon \Omega \to \RR^+ \setminus\{0\}$, for $i \notin A$ and a function $\Psi\colon \Omega \to \RR$.

We make the following assumptions on the space $\Omega$ and the target measure $\mu$: 
\begin{enumerate}
    \item For every $i \in K$, we assume that each boundary $\partial \Omega_i := \overline \Omega_i \setminus \Omega_i$ is a subset of  the union of $(d-1)$-dimensional Riemannian manifolds. Formally, $\partial \Omega_i \subset \cup_{j \in K_i} \Gamma_{j}$, where each $\Gamma_{j}$ is a Riemannian manifold equipped with the metric tensor $g_{j}$ and $|K_i|<\infty$ induced by the imbedding of $\Gamma_j$ in $\R^d$. We define the induced Riemann measure $\partial \lambda$ in each boundary $\Gamma_{j}$ as
    \begin{equation}
    \label{eq: induced riemann measure}
    \partial \lambda(A) := \int_A \sqrt{|\det(g_{j})|}\dd x, \qquad \forall A \in \cB(\Gamma_{j})
    \end{equation}
    where $x$ denote local $(d-1)$-dimensional coordinates for  $\Gamma_{j}$ and  $\cB(A)$ denotes the Borel $\sigma$-algebra on $A$.
    \item Define the set of edges and corners of the boundary as
    \[C := \bigcup_{i \in K} \bigcup_{\substack{j.k\in K_i\\j\ne k}}( \partial \Omega_i \cap \Gamma_{j} \cap\Gamma_{k})\]
    and we assume that $\partial\lambda(C) = 0$. 
    \item For every $i\in  K$ and $x\in \partial \Omega_i$ we assume that the limit
    \[
    \Psi(x,i) := \lim_{\substack{y \rightarrow x \\ y \in \Omega_i}} \Psi(y) 
    \]
   exists. In this way, we can extend $\Psi$ on the closure $\overline \Omega_i := \partial \Omega_i \cup \Omega_i$. 
    \item We assume that  $x \mapsto \exp(-\Psi(x))$ is bounded and differentiable in each restriction $\overline \Omega_i, \, i \in K$. Notice that, with this setting, we allow the target density to vanish to 0 when approaching a boundary.

\end{enumerate}

With this setting, for  $x \in \{ \partial \Omega_j \cap \partial \Omega_i, \, i,j \in K\}$ if $\Psi$ is discontinuous at that point, then $\exp(-\Psi(x, i)) \ne \exp(-\Psi(x,j))$ and if, moreover, $x \in \partial \Omega \setminus C$ (it is neither a edge nor a corner point of the boundary), then we can define the outward normal vector at $x$ by $n(x,i)$ and we have that $n(x, i) = -n(x, j)$.

Next, we build PDMP samplers which can target $\mu$, considering initially non-smooth measures with Lebesgue dominating measure (with $A = \{1,2,\dots,d\}$ in equation~\eqref{eq: dominated measure})  In Section~\ref{sec: sticky extensions}, we extend the setting to the general case.

\subsection{Building blocks of PDMPs with boundaries}\label{sec: building block of PDMPs}
The setting presented here follows closely \textcite[Chapter 2]{Davis1993}. Define the set \[\partial \Omega_i^\dagger := \{x \in \partial \Omega_i \colon \exp(-\Psi(x,i))>0\}\]
corresponding to the boundary points with non-vanishing density. 
Denote the space of possible velocities of the process by $\cV \subset \RR^d$ and, for all $i \in K$, the augmented spaces $E^\circ_i := \Omega_i \times \cV$ and $\partial E^\circ_i := \partial \Omega^\dagger_i \times \cV$, with elements given by the tuple $z  := (x,v)$.

Here, the deterministic flow of the process, its state space and the boundary of the process are jointly defined: let $\phi: E\times \RR \to E$ be the deterministic flow of the process where $E := \cup_{i \in K} E_i$ is the state space of the process with  
 \[
E_i := \left(E_i^\circ \sqcup \partial E^-_i \right)
\]
and
\begin{equation}
      \label{eq: non-active boundary}
      \partial E^-_i := \{z \in \partial E^\circ_i \colon z = \phi(-t, z'), \text{ for some } z' \in E^\circ_i, t>0\}.
\end{equation}
Let $\partial E^+ = \cup_{i \in K} \partial E_i^+$ be the boundary of the process with
 \begin{equation}
     \label{eq: active boundary}
           \partial E_i^+ := \{z \in \partial E^\circ_i \colon z = \phi(t, z'), \text{ for some } z' \in E^\circ_i, t>0\}
 \end{equation}

 Note that $\partial E^+$ is the set of position and velocity for which the process reaches the boundary, while the elements of $\partial E^- := \cup_{i \in K} \partial E^-_i$ are those for which the process leaves the boundary. 
 Furthermore, the set $\bigcup_{i\in K}(\partial \Omega_i \setminus \partial^\dagger \Omega_i) \times \cV$  is neither part of the state space, nor of the boundary. This is not problematic, as the PDMP samplers we consider never reach this set (see Proposition~\ref{rmk: prob to reflect before Psi vanishes}).

 PDMPs are characterized by a finite collection of random events and deterministic dynamics in between those events as follows: \begin{itemize}
\item  
 The deterministic flow of the process $\phi\colon E\times \RR \to E$ takes the differential form
\begin{equation}
    \label{eq: deterministic dynamics 1}
    \frac{\dd \phi(z_0, t)}{\dd t} = (v s(\phi_x(z_0, t)), 0), \quad \phi(z_0, 0) = z_0
\end{equation}
with $\phi_{x}$ being the position component of $\phi$ and for a speed-up function $s \colon \Omega \to \RR^+ \setminus \{0\}$. Similarly to what has been done for $\Psi$, we extend $s$ to the boundary by assuming that, for every $x \in \partial \Omega^\dagger_i, \, i \in K,$ $s(x,j) = \lim_{y \to x, y \in \Omega_i}$ does not depend on $y$ and we assume that $s$ is differentiable in each set $\overline \Omega_i, \, i \in K$.

The dynamics in equation~\eqref{eq: deterministic dynamics 1} generate straight lines in the position coordinate with speed proportional to the function $s \colon \Omega \to \RR^+$.

Generalization to dynamics other than straight lines are possible, for example one might consider PDMPs with Hamiltonian dynamics invariant to Gaussian measures (\cite{bierkens2020boomerang}).

\item  A collection of random event times $\tau_1,\tau_2,\dots$ determines the times where the process changes velocity component. These are computed recursively and are determined by a rate function $\lambda \colon E \to \RR^+$. The
first event time $\tau_1$  of the process starting at $z_0 \in E$ coincides with the first event time of an \emph{inhomogeneous Poisson process} with rate $t \mapsto \lambda(\phi(z_0, t))$ and therefore satisfies
\begin{equation}
    \label{eq: prob of first event time}
    \PP(\tau_1 > t) = \exp\left(-\int_0^t \lambda(\phi(z_0, u)) \dd u\right).
\end{equation}
Hereafter, we write $\tau_1 \sim \text{IPP}(t \mapsto \lambda(\phi(z_0, t)))$. 
 \item Two kernels $\cQ_{\partial E}, \cQ_{E}$ determine the behaviour of the process respectively when approaching the boundary and at random event times.   The two kernels are combined in a Markov kernel $\cQ\colon E \sqcup \partial E^+\times \cB(E) \to [0,1]$ defined as
    \[
    \cQ(z,\cdot) := 
        \begin{cases} 
            \cQ_{E}(z,\cdot) & z \in E \\
            \cQ_{\partial E}(z,\cdot) & z \in \partial E^+.
        \end{cases}
    \]
\end{itemize}

 Given the components $(\phi, \lambda, \cQ)$, a sample path of a PDMP can be simulated by iterating Algorithm~\ref{alg: algorithm abstract}, which compute the first random events of the sample path and returning a finite collection of points (hereafter denoted the skeleton of the process) from which we can deterministically interpolate the whole path. 
\begin{algorithm}
  \caption{Recursive construction of a sample path \label{alg: algorithm abstract}}
   For a given time $t$ and position $z \in E$:
  \begin{itemize}
    \item Compute $\tau^\star := \inf_{u>0} \{\phi(z, u) \in \partial E\}$.
    \item Draw $\tau \sim \text{IPP}(s \mapsto \lambda(\phi(z, s)))$.
    \item If $\tau^\star  < \tau$: save the tuple $(t + \tau^\star, z' = \phi(z, \tau^\star))$, draw $z'' \sim \cQ_{\partial E}(z', \cdot)$. Save and return $(t + \tau^\star, z'')$.
    \item Otherwise (if $\tau < \tau^\star$): draw $z' \sim \cQ_{E}(\phi(z, \tau), \cdot)$. Save and return $(t + \tau, z')$. 
 \end{itemize}
\end{algorithm}

 Similar to \textcite{vasdekis2021speed}, the speed-up function $s(x)$ is allowed to generate exploding dynamics, that is dynamics for which the process escape to infinity in finite time. This will not be problematic, as long as we make sure that a random event switches velocity before the process escapes to infinity or reaches boundaries with vanishing density. This condition is reflected by the following assumption:
\begin{assumption}\label{ass: speed growth condition}\emph{(Speed growth condition)} Let the speed-up function $s$ satisfy \begin{itemize}
    \item $ \lim_{\|x\| \to \infty} \exp(-\Psi(x))s(x) = 0,$
    \item $\lim_{\substack{x \to y\\x \in \Omega}}\exp(-\Psi(x))s(x) = 0,
    \quad y \in \partial \Omega \setminus \partial \Omega^\dagger. $
\end{itemize}
\end{assumption}
Next, we exclude pathological behaviours of the process when hitting the boundary of the state space i.e. the case where the process hits the boundary an infinite number of times in a finite time horizon. These examples are difficult to find in real-word statistical applications and have been studied, independently from the development of PDMP samplers, for example for deterministic dynamical systems with constant velocity and deterministic elastic collisions at the boundary (which are naturally connected with the Bouncy particle sampler, see equation~\eqref{eq: bps relfection at the boundary} below) on a convex subset of $\RR^2$ (a billiard table), see \textcite{halpern1977strange} and references therein for details. 

Fix the initial point $z \in E$, let $\tau_1 := \inf\{t \colon \phi(z, t) \in \partial E^+\}$, and recursively $\tau_i := \inf\{t \colon \phi(z^{(i-1)}, t) \in \partial E^+\}$, with $z^{(i)} \sim \cQ_{\partial E}(\phi(z^{(i-1)}, \tau_i), \cdot), \, z^{(0)} = z$, be the sequence of hitting times to the boundary, with the convention that if $\tau_i = \infty$ then $\tau_j = \infty$ for all $j\ge i$. 
\begin{assumption} For all initial points,  
$
\PP(\lim_{n \to \infty} \tau_n < \infty) = 0.
$
\end{assumption}

    Denote  by $C^1_c(E)$ the space of compactly supported functions on $E$ which are continuously differentiable in their first argument. 
    A necessary condition for PDMPs to  target the measure $\mu \otimes \rho$, where $\rho$ is the marginal invariant measure of the velocity component while $\mu$ is the target measure we are interested to sample from, is that, for functions in 
\begin{align}
    \cA := \{f \in C^1_c(E)\colon&\, t\mapsto f(\phi(z,t)) \text{ is absolutely continuous } \forall z \in E; \nonumber \\
    &   f(z)= \int_{\partial E^-} f(z') \cQ_{\partial E}(z, \dd z'), \quad \forall z \in \partial E^+\} \label{eq: boundary condition domain of the generator}
\end{align}
the following equality holds 
\begin{equation}\label{eq: int generator f relative to mu and rho is zero}
\int_E \cL f \, \dd (\mu\otimes\rho) = 0
\end{equation}
where $\cL$ is the extended generator of the process. For PDMPs, $\cL$ is known and its expression is given in Appendix~\ref{app: generator of PDMPs}.

Remark~\ref{rmk: equvalence with weighted PDMP} below sheds light on the connection between the speed-up function and importance weights, with $s^{-1}$ being the importance function (see \cite[Section 3.3]{robert1999monte}) and therefore may used as a guideline for the design of the speed-up function in relation to the asymptotic variance of the Monte Carlo estimator and for an efficient implementation of PDMPs, with the consideration that, once the speed-up function has been chosen, the weights of the path can be computed retrospectively in a post-process stage.

\begin{remark}\label{rmk: equvalence with weighted PDMP}
\emph{(Equivalence between PDMPs with speed-up function and ordinary PDMPs with weighted path)}
 Consider a ordinary PDMP $Z'_t := (X'_t, V'_t)$ with final clock $T'$ and without speed-up function, targeting a measure with unnormalized density
\[
\pi^*(x) := \exp\left(-\Psi(x)\right)s(x),\quad x \in \Omega
\]
and a PDMP $Z_t := (X_t, V_t)$ with final clock $T := \int_0^{T^\star} s(X_t)^{-1} \dd t $ and speed-up function $s$ targeting a measure with unnormalized density $\exp(-\Psi(x))$. 
Then, the equivalence
\[
\frac{1}{T'}\int_0^{T'} \frac{1}{s(X'_t)} f(Z'_t) \dd t
=
\frac{1}{T}\int_0^{T} f(Z_t) \dd t
\]
holds almost surely for all $f\colon E \to \RR$ with $\int f \dd (\mu \otimes\rho) < \infty$.
\end{remark}

Next, we impose sufficient conditions for \eqref{eq: int generator f relative to mu and rho is zero} to hold. In particular, we will see that conditions on different components of PDMPs serve to target different components of the measure in equation~\eqref{eq: general form target measure}. In Section~\ref{subsec: conditions on the interior}, we review standard conditions on
$(\lambda, s, \cQ_{E})$ imposed for ordinary PDMPs (see for example \cite{vasdekis2021speed}) which allow the process to target the differentiable part of the density within each set $\Omega_i, \, i\in K$, while in Section~\ref{subsec: conditions at the boundary}, new 
conditions on $(s, \cQ_{\partial E})$ are imposed which guarantee that the process targets the measure $\mu$ when approaching each boundary $\partial \Omega_i,\, i \in K$ and allow to define teleportation portals and to target measures which are discontinuous over the boundary. In Section~\ref{sec: ergodicity pdmp with boundaries} we analyse the limiting behaviour of the process.

\subsection{Review of sufficient conditions on the interior}\label{subsec: conditions on the interior}   

For our PDMP samplers the kernel $\cQ_{E}$ acts by only changing the velocity component, while leaving the position unchanged. Hence, with abuse of notation, we let $\cQ_{E}\colon E\times \cB(\cV) \to [0,1]$ be a kernel acting only on the velocity component. Throughout, we distinguish between two different classes of events: \emph{reflections} and \emph{refreshments}, which are defined by rates and kernels $(\lambda_{\mathrm{b}}, \cQ_{E, \mathrm b}) $ and $(\lambda_{\mathrm{r}}, \cQ_{E, \mathrm r} )$, respectively. The former ensures that the process targets the right measure while the latter ensures ergodicity of the process. Then, for $z \in E$,  $\lambda(z) = \lambda_{\mathrm{r}}(z) + \lambda_{\mathrm{b}}(z)$ and 
\[
  \cQ_{E}(z,\cdot) := \frac{\lambda_{\mathrm{r}}(z)}{\lambda_{\mathrm{r}}(z) + \lambda_{\mathrm{b}}(z)}  \cQ_{E, \mathrm r}(z, \cdot) + \frac{\lambda_{\mathrm{b}}(z)}{\lambda_{\mathrm{r}}(z) + \lambda_{\mathrm{b}}(z)}  \cQ_{E, \mathrm b}(z, \cdot).
\]
Then, for a fixed point $z \in E$, the first random event of a PDMP with rate $\lambda$ and transition kernel $\cQ$ can be simulated by \emph{superposition}: the first event time is given by $\tau  = \min(\tau_\rr,\tau_\bb)$ where $\tau_i \sim \text{IPP}(s \mapsto \lambda_i(\phi(z, s)))$ and the new velocity is drawn as $v' \sim \cQ_\rr(\phi(z, \tau), \cdot)$, if $\tau_\rr < \tau_\bb$ and $v' \sim \cQ_\bb(\phi(z, \tau), \cdot)$ otherwise.

Next, we makes assumptions on both the refreshment and the reflection events. Recall that the desired target measure of the process takes the form $(\mu\otimes\rho)(\dd x, \dd v) := C\exp(-\Psi(x))\dd x\rho(\dd v)$, for some constant of normalization $C$ and invariant measure $\rho$ of the velocity component.
\begin{assumption}\label{ass: conditions on the refreshments} \emph{(Conditions on refreshments)} Let  $0< \lambda_{\mathrm{r}}(x,v)  < \infty$ be a function which is constant on its second argument. Furthermore, let $\cQ_{\mathrm {E, r}}$ be invariant to $\rho$, i.e. \[\int_{v \in \cV} \rho(\dd v)\cQ_{E, \mathrm{r}}((x,v), \dd v') = \rho(\dd v'), \quad \forall x \in \{y \in \RR^d \colon (y,v) \in E\}\]
and such that $\cQ_\rr(z, \cdot)$ is supported in $\cV$, for all $z$.
\end{assumption} 
It is customary for PDMP samplers to set $\cQ_{E, \mathrm{r}}((x,v), \cdot) = \rho(\cdot)$ and a fixed rate $\lambda_{\mathrm{r}}(x,v) = c \ge 0$.

\begin{assumption}\label{ass: conditions on the interior}\emph{(Conditions on reflections)}
    For all $(x,v) \in E$ and for a PDMP with speed-up function $s(x)$, let  $\lambda_{\mathrm{b}}\colon E \to \RR^+$ satisfy
    \[
    \lambda_{\mathrm{b}}(x,v) - \lambda_{\mathrm{b}}(x,-v)= \langle v, A(x)  \rangle   
    \]
    with 
    \[
    A(x) = s(x) \nabla \Psi(x) - \nabla s(x).
    \]
    Let $\cQ_{E, \mathrm{b}}$ satisfy
\begin{equation}
   \label{eq: condition of kernel on the interior}
 \int_{v \in \cV}\rho(\dd v) \lambda(x,v)  \cQ_{E, \mathrm{b}}((x, v), \dd v')    = \lambda(x,-v') \rho(\dd v'), \quad \forall x \in \{y \in \RR^d \colon (y,v) \in E\}.
\end{equation}

With the following Proposition, we informally justify the choice made in Section~\ref{sec: building block of PDMPs} of not considering the set $\partial \Omega\setminus \partial \Omega^\dagger \times \cV$ (set of position and velocity where the position is on the boundary with vanishing density) as either part of the state space or boundary of the process.
\begin{proposition}\label{rmk: prob to reflect before Psi vanishes}\emph{(Probability to reflect before hitting a boundary with vanishing density)}
Let $\lambda_{\bb}$ and $s$ be as in Assumptions~\ref{ass: speed growth condition}-\ref{ass: conditions on the interior}. For every $z \in E$ such that $y = \phi_x(z, \tau^\star) \in \partial \Omega\setminus \partial \Omega^\dagger, \, \tau^\star>0$, Let $\tau_\bb \sim \text{IPP}(t\mapsto \lambda_\bb(\phi(z, t)))$ be the random reflection time. Then,
$
\PP(\tau_\bb < \tau^\star) = 1.
$
\end{proposition}
\begin{proof}
By Assumption~\ref{ass: conditions on the interior} we have that 
\[
\lambda_\bb(x,v) \ge \max(\langle v, \nabla \Psi(x)s(x) - \nabla s(x)
\rangle, 0).
\]
Then, by 
 Assumption~\ref{ass: speed growth condition} we have that
\[
\PP(\tau_\bb < \tau^\star) \ge 1 -  \lim_{u \uparrow \tau^\star} \exp\left(-\int_0^{u} \langle v, 
\nabla \Psi(\phi_x(z, s))s(\phi_x(z, s)) - \nabla s(\phi_x(z, s))\rangle  \dd s\right) = 1 -  \frac{\exp(-\Psi(y))s(y)}{\exp(-\Psi(x))s(x)} = 1.
\]
\end{proof}

\end{assumption}
The next proposition shows that if we make the assumptions above and compute the left-hand side of equation~\eqref{eq: int generator f relative to mu and rho is zero}, we are left with an integral over $\partial E := \partial E^+ \sqcup \partial E^-$.
\begin{proposition}\label{prop: integral of generator with standard assumptions}
Consider a PDMP sampler satisfying Assumption~\ref{ass: conditions on the refreshments}-\ref{ass: conditions on the interior}. Then
\begin{equation}\label{eq: integral of generator with standard assumptions}
\int_E \cL f \dd (\mu\otimes \rho) =
        \int_{(x,v) \in\partial E} f(x,v)\langle n(x), v\rangle  s(x) \exp(-\Psi(x)) \partial \lambda (\dd x) \rho(\dd v), \quad f \in \cA,
\end{equation}
where $\cL$ is the extended generator of the process, $\cA$ is given in~\eqref{eq: boundary condition domain of the generator} and $\partial \lambda$ is the induced Riemann measure on the boundary introduced in equation~\eqref{eq: induced riemann measure}.
\end{proposition}
\begin{proof}
See Appendix~\ref{app: proof proposition condiotions on the interior}.
\end{proof}
Notice that $(x,v) \mapsto \langle n(x), v\rangle  s(x) \exp(-\Psi(x))$ is not defined in the set of edges of corners of the boundary $C \times \cV$ as the normal vector is not defined for corner points. However, this is not problematic since, by the assumptions made in Section~\ref{sec: pdmps samplers with boundaries}, we have $(\partial \lambda\otimes \rho)(C \times \cV) = 0$. 

For a position $z \in E$, simulating the next random event time as in equation~\eqref{eq: prob of first event time} can be challenging, depending on the measure $\mu$ and speed-up function $s$. A standard technique is to find a function  $ \bar \lambda(z, u) > \lambda(\phi(z, u))$ for all $u >0$ such that $\tau \sim \text{IPP}(u \mapsto \bar\lambda(z,u))$ can be computed. Then  $\tau$ is a event time with probability $\lambda(\phi(z, \tau))/\bar \lambda(z, \tau)$. This scheme is referred  as to thinning. If the upper bound is not tight to the Poisson rate $\lambda$, the procedure induce extra computational costs that can deteriorate the performance of the sampler. Several numerical schemes have been recently proposed trying to address this issue, see for example \cite{pagani2020nuzz}, \textcite{corbella2022automatic}, \textcite{bertazzi2022adaptive} and \textcite{sutton2021concave}.

Next, we will show that the right hand-side of equation~\eqref{eq: integral of generator with standard assumptions} can be made equal to 0 by detailing the behaviour of PDMPs at the boundary $\partial E^+$.

\subsection{Sufficient conditions at the boundary}\label{subsec: conditions at the boundary}
Here, we state a fairly general assumption for the kernel $\cQ_{\partial E}$ (Assumption~\ref{ass: sdb at the boundary}) which guarantees that the process targets the measure $\mu$ when it approaches the boundary. This assumption on the boundary, together with Assumptions~\ref{ass: speed growth condition}-\ref{ass: conditions on the refreshments}-\ref{ass: conditions on the interior} allows to determine the invariant measure of the PDMPs (Theorem~\ref{thm: main}). Equation~\eqref{eq: Q at the boundary} and Algorithm~\ref{alg: algorithm boundary} gives a concrete example of a transition kernel at the boundary which will be used in the applications of Section~\ref{sec: applications}.

Recall that a map $\cS\colon \cX \to \cX$ is an involution if 
$\cS \circ \cS = I$, where $I$ stands for the identity map. We now give an important definition, followed by the main assumption made at the boundary:
\begin{def}\emph{(Skew detailed balance condition)} For an involution $\cS : \cX \to \cX$, a kernel $\cQ\colon (\cX, \cB(\cX)) \to [0,1]$ satisfies the skew detailed balance condition relative to a measure $\mu$ on $\cX$, if 
\[
\cQ(z, \dd z')\mu(\dd z) =  \cQ(\cS^{-1}(z'),\cS^{-1}(\dd z)) \mu(\cS^{-1}(\dd z')).
\]
\end{def}

\begin{assumption}\label{ass: sdb at the boundary} Let 
$\cQ_{\partial E}\colon \partial E^+\times \cB(\partial E^-) \to [0,1]$ satisfy the skew detailed balance condition relative to the signed measure $ \nu(\dd x, \dd v) = \langle n(x), v \rangle s(x)\partial \lambda(\dd x) \rho(\dd v)$ defined on $\partial E$ with involution $\cS(x,v) = (x,-v)$.
\end{assumption}
Notice that, contrary to the standard applications of skew detailed balance, here $\nu$ is a signed measure and $\nu(\dd z) = -  \nu(\cS^{-1}(\dd z))$. The use of the skew detailed balance condition relative to a signed measure allows to handle the symmetry arising between exit boundaries $\partial E^-$ and entrance boundaries $\partial E^+$, which is made explicit in the proof of Theorem~\ref{thm: main} below.

\begin{theorem}\label{thm: main}
    Consider a PDMP sampler satisfying Assumption~\ref{ass: conditions on the refreshments}-\ref{ass: conditions on the interior}-\ref{ass: sdb at the boundary}. Then $\mu \otimes \rho$ is a stationary measure of the process.
\end{theorem}

\begin{proof}
If $\mu \otimes \rho$ is a stationary measure of the PDMP, then we must have that $\int \cL f \dd (\mu \otimes \rho) = 0, \, f \in \cA$.
By Proposition~\ref{prop: integral of generator with standard assumptions} and the boundary condition in~\eqref{eq: boundary condition domain of the generator}, we have that
\begin{align}
    \int_E \cL f \dd (\mu\otimes \rho) &=
    \int_{z \in \partial E^+} \int_{p \in \partial E^-} \cQ_{\partial E}(z, \dd p ) f(p) \nu(\dd z) \nonumber\\
    &\qquad + \int_{\partial E^-} f(z) \nu(\dd  z) \nonumber\\
    &= \int_{p \in \partial E-} f(p) \int_{z \in\partial E^+} \cQ_{\partial E}(z, \dd p) \nu(\dd z) \label{eq: proof sdb 1}\\
    &\qquad + \int_{\partial E^-} f(z)\nu(\dd z)\nonumber\\
    &=  \int_{\partial E^-}f(p)   \nu(\cS(\dd p)) + \int_{\partial E^-} f(z)\nu(\dd z)= 0 \label{eq: proof sdb 2}
\end{align}

where  in~\eqref{eq: proof sdb 1} we applied Fubini for interchanging integrals and in~\eqref{eq: proof sdb 2} we used  Assumption~\ref{ass: sdb at the boundary}.
\end{proof}

 For simplicity, in this article, we focus on specific transition kernels at the boundary which satisfy Assumption~\ref{ass: sdb at the boundary} and will be used for the two main applications in Section~\ref{sec: applications}.
 The transition kernels considered take the form
 \begin{equation}
    \label{eq: Q at the boundary}
    \cQ_{\partial E}((x,v), \dd (y, w)) := 
 \begin{cases}
    \delta_x(\dd y)\delta_{-v}(\dd w), & x \in C, \\
     \begin{split}
     &\alpha(x, y)\cT(x, \dd y)\cR_1((x,v), y, \dd w) \\ &\qquad + \left(1 - \int \alpha(x, y)\cT(x, \dd y)\right) \delta_x(\dd y) \cR_2((x, v),\dd w),
      \end{split} 
     & \text{otherwise}.
 \end{cases}
\end{equation}
In this setting, upon hitting a boundary $(x,v)\in \partial E^+$, if this is an edge or corner point of the boundary, then the process reverts its velocity component, otherwise it jumps to $(y, w)$ with $y \sim \cT(x, \cdot)$ and $w \sim \cR_1((x,v),y, \cdot)$ with probability $\alpha(x,y)$ and reflects at the boundary otherwise, by setting a new velocity according to $q \sim \cR_2((x,v),\cdot)$, see Algorithm~\ref{alg: algorithm boundary} for its implementation. We now specify in details each term in equation~\eqref{eq: Q at the boundary}.

The kernel $\cT\colon \partial \Omega \times\cB(\partial \Omega) \to [0,1]$ acts on the boundary of the position component of the process and such that for every $x \in \partial \Omega$,  $\cT(x, \cdot)$ is not supported on the set of corners $C$ and it is absolutely continuous with respect to $\cT(y, \cdot)$, almost surely for every $y \sim \cT(x, \cdot)$. We set $\alpha(x,y) := \min(1, R(x,y))$ where $R(\cdot, \cdot)$ is the Radon-Nikodym derivative on the product space $(\partial \Omega\times \partial \Omega, \cB(\partial \Omega\times\partial \Omega))$ defined as
 \begin{equation}\label{eq: randon nikodim derivative acc-rej}
 R(x,y) :=  \frac{\nu(\dd y)  \cT(y, \dd x)}{\nu(\dd x)\cT(x, \dd y)}
 \end{equation}
where $\nu( \dd x) = \exp(-\Psi(x))s(x) \dd x$.
 Finally, for every $(x,v) \in \partial E^+$ and $y \in \cT(x, \cdot)$, we define two kernels acting on the velocity components:  $\cR_1 \colon \{(x,v)\in \partial E^+\} \times \{y \in \partial \Omega\} \times \cB(\partial \cV^-(y)) \to [0,1]$ and $\cR_2 \colon \{(x,v)\in \partial E^+\} \times \cB(\partial \cV^-(x)) \to [0,1]$, with
 $
\partial \cV^-(x) = \{v \in \cV \colon (x,v) \in \partial E^-\},
$
which, for every $(x,v) \in \partial E^+   \setminus C \times \cV$ and every $y \in \partial \Omega^\dagger \setminus C$ satisfy 
\[
\int_{v \in \partial \cV^{-}(x)} \cR_1((x, v), y, \dd w) \langle n(x), v \rangle \rho(\dd v) = -\langle n(y), w \rangle \rho(\dd w),
\]
and 
\[
\int_{v \in \partial \cV^{-}(x)} \cR_2((x, v), \dd w) \langle n(x), v \rangle \rho(\dd v) = -\langle n(x), w\rangle \rho(\dd w). 
\]

\begin{algorithm}
  \caption{Behaviour of PDMPs at the Boundary \label{alg: algorithm boundary}}
   For $(x,v) \in \partial E^+$:
  \begin{itemize}
    \item If $x \in C$:
    \begin{itemize}
        \item Set the new state $(x', v') = (x,-v)$. 
    \end{itemize}
    \item Otherwise:
    \begin{itemize}
     \item Propose a point $y\in \partial \Omega$ as $y \sim \cT(x, \cdot)$. 
     \item Simulate $u \sim \text{Unif}([0,1])$.
     \item If 
     $
      \alpha(x, y) > u,
     $
     set the new state to  $(x', v') = (y, w)$ with $w \sim \cR_1((x,v), y, \cdot)$. \item Otherwise set the state $(x', v')= (x, q)$ with $q \sim \cR_2((x,v), \cdot)$. 
     \end{itemize}
    \item return the new state $(x', v')$.
 \end{itemize}
\end{algorithm}

The following remarks may serve as a guidance for the choice of the boundaries and speed-up function $s$:
\begin{remark}\label{rmk: sppedup}
The speed-up function $s$ can be chosen to be constant in each restriction $\Omega_i,\, i \in K$ and can be tuned in order to reduce the probability to reflect at the boundary (see equation~\eqref{eq: randon nikodim derivative acc-rej}) and, as long as Assumption~\ref{ass: speed growth condition} is satisfied, it can be calibrated in each restriction $\Omega_i,\, i \in K$ to completely off-set that probability so that the process always crosses the boundary. In this case, there is no need to specify a reflection rule of the velocity at the boundary. This is key whenever there is no good choice to reflect the velocity at the boundary rather than flipping completely the velocity vector, hence generating undesirable \emph{back-tracking} effects of the underlying Markov process.
\end{remark}

\begin{remark}\label{rmk: sppedup 2} The boundaries $\partial \Omega_i,\, i\in K$ and the speed-up function $s$ can be set 
such that $s$ is constant in each restriction and it approximates 
a continuous speed-up function for which the deterministic dynamics cannot be integrated analytically, this way, extending the choice of continuous functions considered in \textcite{vasdekis2021speed}. 
\end{remark}

\subsection{Ergodicity of PDMPs with boundaries}\label{sec: ergodicity pdmp with boundaries}
Here, we analyse the limiting measure of the PDMP samplers. Ergodicity for PDMPs with boundary conditions has only been established for the specific case of hard-sphere models in \textcite{monemvassitis2022pdmp}.

 In the following, we follow the strategy undertaken in \textcite{2015arXiv151002451B} for characterizing the limiting measure of the PDMP samplers $(Z_t)_{t>0}$ with piecewise linear dynamics (Zig-Zag sampler and Bouncy particle sampler). 
The strategy consist of showing that the process is \emph{open set irreducible} (Proposition~\ref{prop: accessibility}) and that a minorization condition holds (Proposition~\ref{prop: minorization condition}). As shown in \textcite{2015arXiv151002451B}, these results together with  Theorem~\ref{thm: main} guarantees that the process is ergodic (Theorem~\ref{thm: ergodicity}) with unique invariant measure equal to $\mu\otimes\rho$. We make additional assumptions on the transition kernels $\cQ_{\partial E}$ at the boundary to make sure that the process is able to jump between any two sets $E_i, \,E_j, \, i,j \in K$. The main arguments mainly rely on assuming that the process can refresh its momentum component at any time (Assumption~\ref{ass: conditions on the refreshments}). Generalization of this argument may be possible, see for example \textcite{bierkens2019ergodicity} which showed that, under some more stringent conditions on $\Psi$, the Zig-Zag sampler in a  space without boundaries is ergodic without refreshments of the velocity component. We present here the most important results that we derived, together with a road map of the proofs, while deferring to the appendix many intermediate results and all the proofs.

We first define an admissible path which connects two points in $E_i$ for which the position and velocity follows the deterministic flow $\phi$ and the velocity changes on a  finite collection of times:
\begin{definition}\label{def: admissible paths}\emph{(Admissible path)} A function $(z(t) \in E)_{0\le t\le T}$ is said to be an admissible path in $E_i$ if there is a finite collection of tuple $(\tau^{(i)}>0, z^{(i)} = (x^{(i)}, v^{(i)}) \in E_i)_{i=0,1,\dots,N}$ with $N<\infty$, such that $z(0) = z^{(0)}, \, z(T) = z^{(N)}$ and $x^{(i+1)} = \phi_x(z^{(i)}, \tau^{(i)})$.
\end{definition}

The first proposition is a straightforward consequence that every two points in a connected open subset of $\RR^d$ are \emph{path-connected} and in particular they can be connected with a piecewise linear path (Lemma~\ref{prop: connected to path} in Appendix~\ref{app: open set irreducible}). Recall that, for $i \in K$, $E_i = E_i^\circ \sqcup \partial E^-_i$ where $E^\circ_i$ is open and $\partial E_i^-$ contains those points $y = \phi(z, -t)$ for some $z \in E^\circ_i$ and some $t>0$ obtained by applying the deterministic flow backward in time, with starting point in $E_i$ (as such, the set $\partial E^-_i$ is clearly path connected with the interior $E^\circ_i$).

\begin{proposition}\label{prop: exists a path connecting points}
For every $i \in K$, for every two values $z \in E_i, \, z' \in E^\circ_i$, there exists an admissible path $(z(t))_{0\le t\le T}$ 
with  $z(0) = z, \, z(T) = z'$. 
\end{proposition}

Next, we assume that the transition kernel at the boundary allows the process to jump between every two regions $E_i,\,E_j, \, i,j \in K$.

\begin{definition}\label{def: direct neighbour}\emph{(Direct neighbouring sets)} A set $E_i$ is said to be a  direct neighbour of a set  $E_j$ with $i,j \in K$  if there exists a open subset $A\subset \partial E^{+}_i$, a set $B \in \cB(\partial E_j^-)$, and $\epsilon > 0$ such that  
\[
\cQ_{\partial E}(z, B) \ge \epsilon, \quad \forall z \in A. 
\] 
 
In this case, we write $E_i \curvearrowright E_j$.
\end{definition}

In other words, $E_i \curvearrowright E_j$ 
if there is a region $A \subset \partial E_i^+$  that allows the process to jump on $\partial E^-_j$ with a positive probability uniformly for every $z \in A$.
\begin{example}\emph{(Deterministic jumps)} Consider a  continuous function $\kappa\colon \partial E_i^+ \to \partial E_j^-$ and a transition kernel $\cQ_{\partial E}(z, \cdot) = \delta_{\kappa(z)}(\cdot), \, \forall z \in \partial E_i^+$. Then, $E_i \curvearrowright E_j$ with $A \subset \partial E^+_i$ be any open set and $B =\{ \kappa(z) \in \partial E^-_j \colon z \in A\}$ and $\epsilon = 1$ in Definition~\ref{def: direct neighbour}.
\end{example}
\begin{definition}
\emph{(Neighbouring sets)}  A set $E_i$ is said to be  a neighbour of a set  $E_j$ with $i,j \in K$ if there exists a sequence $i = {k_1}, {k_2},\dots,{k_N} = j$ with elements in $K$, such that $E_{i_k}\curvearrowright E_{i_{k+1}}$, for $k=1,2,\dots,N-1$. In this case, we write $E_i\dashrightarrow E_j$.
\end{definition}

\begin{assumption}\label{ass: connectivity}
$E_i \dashrightarrow E_j$, for all $i,j \in K$. 
\end{assumption}

 Let $P_t \colon E \times \cB(E) \to [0,1], \, t\ge 0$ be the Markov semigroup  associated with the PDMP sampler (here we consider both Zig-Zag and Bouncy Particle sampler).
 
\begin{proposition}\label{prop: accessibility} 
Suppose Assumptions~\ref{ass: conditions on the refreshments}-\ref{ass: conditions on the interior}-\ref{ass: sdb at the boundary}-\ref{ass: connectivity} hold. Then, for every $z \in E$, for every open and non-empty set $A \subset E$, and for every $\Delta>0$, there exists an integer $n>0$ such that
\[
P_{n\Delta } (z, A) >0. 
\]
\end{proposition}
\begin{proof}
    See Appendix~\ref{app: open set irreducible}.
\end{proof}
The proof in Appendix~\ref{app: open set irreducible} is given by several intermediate steps that we informally list below:
\begin{itemize}
    \item By Assumption~\ref{ass: connectivity}, there is a sequence of sets $E_{k_1} \curvearrowright E_{k_2}\dots\curvearrowright E_{k_n}$, with $z \in E_{k_1}$, $A \subset E_{k_n}$ and with sets $A_{i}\subset \partial E_{k_i}^+ \setminus (C\times \cV)$ and $B_{i+1} \in \cB(\partial E_{k_{i+1}}^-)$ for $i = 1,2,\dots,n-1$ as in Definition~\ref{def: direct neighbour}.
    \item Denote the sets $B_{0} = \{z\}, \, A_{n} = A$ and take any points $z^{i}_B \in B_{i}, z^{i+1}_A \in A_{i+1}$ for $i = 0,1,\dots,n-1$. By Proposition~\ref{prop: exists a path connecting points}, 
    there is an admissible path  $z(t)_{0\le t\le T}$ such that $z(0) = z^{i}_B$ and $\lim_{t \uparrow T }z(t) =  z^{i+1}_A$. 
    \item By using Assumption~\ref{ass: conditions on the refreshments}-\ref{ass: conditions on the interior}-\ref{ass: sdb at the boundary}-\ref{ass: connectivity}, we construct a process which reaches the set $A$ in finite time with positive probability. This is obtained by controlling the set $\cE$ of random events of PDMPs (reflection events, refreshment events and jumps at the boundary) and by showing that this set has positive probability. 
    \item Finally, we confine the process  for any time $\Delta>0$ in the position component inside a ball $B_r(x)$, for any $x \in \Omega$ along its trajectory.
\end{itemize}
The second ingredient used to prove ergodicity is the following minorization condition, which is verified for the Bouncy particle sampler: 
\begin{proposition}
    \label{prop: minorization condition}
    \emph{(Minorization condition)} Suppose Assumptions~\ref{ass: conditions on the refreshments}-\ref{ass: conditions on the interior} hold. There exists $r, \epsilon >0$ such that, for all $\epsilon<t< r$,  there exists a non-empty open set $C \subset E$ and $\delta>0$ such that for all $z \in C$, 
    \[
    P_t(z, A) \ge \delta \Vol(A\cap C), \qquad \forall A \in \cB(E).
    \]
\end{proposition}
A similar result has been proven in \textcite[proof of Lemma 3]{2015arXiv151002451B}, in the Euclidean space without boundary conditions and generalises in our setting, see Appendix~\ref{app: minorization} for more details.

We are now ready to present the most important result of this section:
\begin{theorem}\label{thm: ergodicity}\emph{(Ergodicity of the PDMP samplers with boundary conditions)} Suppose Assumptions~\ref{ass: conditions on the refreshments}-\ref{ass: conditions on the interior}-\ref{ass: sdb at the boundary}-\ref{ass: connectivity} hold. Then, for $\mu\otimes\rho$-almost every initial point $Z(0)= z$, $\mu\otimes\rho$ is the unique invariant measure of the process and, for every $f\colon E \to \RR$, with  $\int f \dd(\mu\otimes\rho) \le \infty$, we have
\[
\lim_{T \to \infty} \frac{1}{T}\int_0^T f(Z(t)) \dd t =  \int f \dd(\mu\otimes\rho), \qquad\text{almost surely.}
\]
\end{theorem}
Theorem~\ref{thm: ergodicity} follows by Proposition~\ref{prop: accessibility}-\ref{prop: minorization condition} and Theorem~\ref{thm: main} (see \cite[proof of Theorem 1]{2015arXiv151002451B}).

\subsection{Extensions with sticky event times}\label{sec: sticky extensions}
Here we extend the framework for targeting measures which have a density relative to mixtures of Lebesgue and atomic components as in equation \eqref{eq: dominated measure}. This extension is based on the theory developed in \textcite{bierkens2021sticky}, for sampling measures which arise in sparse Bayesian problems.

We first introduce some notation. For each $i\in A$ in equation~\eqref{eq: dominated measure}, we define $c^+_i = \lim_{t \downarrow 0, c + vt}$ and 
$c^-_i = \lim_{t \downarrow 0, c - vt}$ and consider the space $\overline\RR = \bigotimes_{i=1}^d\overline\RR^{(i)}$, where $\overline\RR^{(i)} =\RR$ if $i \in A$ and  $\overline\RR^{(i)} =(-\infty, c_i^-] \sqcup [c^+, \infty)$ otherwise. Then, we define
\[
\partial S_i^- = \{ (x,v) \in \overline \RR \times \cV \colon x_i =  c^+, v_i> 0 \sqcup x_i = c^-, v_i < 0\}
\]
which is part of the state space of the process, and
\[
\partial S_i^+ = \{ (x,v) \in \overline 
\RR^d \times \cV \colon x_i =  c^+, v_i< 0 \sqcup x_i = c^-, v_i > 0\}
\]
which is part of the boundary of the process. Denote the transfer mappings $T_i \colon \partial S_i^+ \to \partial S_i^-$ 
by 
\begin{equation}\label{eq: transfer mapping}
T_i(x,v) = \begin{cases}
    x[i \colon c_i^+],v &\text{if } x_i = c^-_i, \,v_i >0, \\
    x[i \colon c_i^-],v &\text{if } x_i = c^+_i, \, v_0< 0. 
\end{cases}
\end{equation}
We let each coordinate process $(x_i,v_i)$ with $i \in A$, sticks at $(c_i,v_i)$, upon hitting the the state, for a time equal to $|v_i|\kappa_i(x)$. After the random time, the coordinate process jumps  on the other side of the half-interval, and  leaves the boundary with its original velocity component, see \textcite[Section 2.1]{bierkens2021sticky} for details. By introducing sticky coordinate components, we allow the process to visit the subspaces given by the restrictions $x_i = c_i$ for $i \in B$, $B$ being any subset of $A$. This was exploited in \textcite{bierkens2021sticky} for sampling target measures arising in Bayesian sparse problems and will be used here in the application considered in section~\ref{sec: epidemiology}.

The mathematical characterization of the sticky behaviour of each coordinate process $(x_i,v_i),\, i \in A$ is made by restricting the class of functions in \eqref{eq: boundary condition domain of the generator} as
\begin{align*}
\tilde \cA &= \{\cA \colon \kappa(x)|v_i| = f(T_i(x,v)) - f(x,v), (x,v) \in \partial S_i^+, i \in A\}.
\end{align*}

For a proof of ergodicity and invariant measure of the process, see \textcite[proof of Theorem 2.2-2.3]{bierkens2021sticky}.  
\begin{remark}\label{rmk: sticking at the boundary}
    If the sticky point $c_i$ lays on the boundary of a region supported by the measure $\mu$, we let the coordinate process leave the boundary on the same side of the real line from which it entered, by negating the velocity component (rather than using the transfer mapping in equation~\eqref{eq: transfer mapping}).
    In this case, in order to preserve the measure $\mu$ in \eqref{eq: general form target measure}-\eqref{eq: dominated measure}, one can verify that the coordinate process must stick at $c_i$ for a random time equal to $\frac{1}{2}|v_i|\kappa_i(x)$ where the factor $\frac{1}{2}$ compensates from the fact that the point $c_i$ can be reached by the coordinate process only from one side of the real line.
\end{remark}

Next, we give two concrete examples of PDMP samplers with boundaries that will be used in the applications of Section~\ref{sec: applications}.
\subsection{Example: \texorpdfstring{$d$}{d}-dimensional Zig-Zag sampler for discontinuous densities}\label{sec: Zig-Zag sampler for discontinuous densities}
The standard $d$-dimensional Zig-Zag sampler without boundaries is a PDMP sampler defined in the augmented space of position and velocity  $\RR^d\times \{-1,+1\}^d$. For the process in the state $z \in E$, the first random reflection time  is given by  $\tau = \min(\tau_1, \tau_2, \dots, \tau_d)$ where for $i= 1,2,\dots,d$,  $\tau_i\sim \text{IPP}(t \mapsto \lambda_{i, \text{b}}(\phi(z, t)))$ and
\begin{equation}
    \label{eq: rates zz}
    \lambda_{i, \text{b}}(x,v) =  \max(v_i(s(x)\partial_{x_i} \Psi(x)  - \partial_{x_i}s(x)), 0).
\end{equation}
At random time $\tau$, the velocity component changes as $v \mapsto v[k; -v_k]$ where $k = \argmin(\tau_1, \tau_2, \dots, \tau_d)$.  For more details on the standard Zig-Zag sampler see \textcite{Bierkens_2018} and \textcite{vasdekis2021speed}.

Here we extend the Zig-Zag process for discontinuous densities. Let split the space in disjoint sets: $\RR^d = \sqcup_{i \in K} \Omega_i$
and let the target measure be discontinuous over the boundary $\partial \Omega_i$. Then, the target measure is of the form of equation~\eqref{eq: general form target measure}, with $A = \{1,2,\dots,d\}$ and the invariant measure of the process is $\mu(\dd x)\times\text{Unif}(\{-1,+1\}^d)$. 
With this setting, for almost every $(x,i) \in \partial \Omega_i, i \in K$ there is a unique $j \in K$ such that the point $(x,j) \in \partial \Omega_j$.
We link points by the function $\kappa \colon \partial \Omega \to \partial \Omega$ with $\kappa(x,i) = (x,j)$. We set $\cT((x,i), \cdot ) = \delta_{\kappa(x,i)}(\cdot)$. As in this case $n(x,i) = -n(\kappa(x,i))$ we set
\[
\cR_1((x,i,v), \kappa(x,i), \dd w) = \delta_v(\dd w)
\]
and $\cR_2((x,i,v), \cdot) =  \delta_{v'}(\cdot)$
with 
\[
v_\ell' =  \begin{cases}
    -v_\ell & n_\ell(x) \ne 0,\\
    v_\ell & n_\ell(x) = 0.
\end{cases}
\]
The process, upon hitting a boundary $(x,i) \in \partial \Omega$, crosses the boundary with probability $\alpha((x,i), \kappa(x,i))$ without changing the velocity component and reflects  otherwise by switching the sign of only the components which are not orthogonal to $n(x,i)$.

\subsection{Example: \texorpdfstring{$d$}{d}-dimensional Bouncy Particle Sampler with teleportation on constrained spaces}\label{sec:  Bouncy Particle sampler with teleportation}
The $d$-dimensional Bouncy Particle Sampler (BPS) is defined in the  space of position and velocity $\RR^{d}\times \RR^d$. The velocity component has a marginal invariant measure equal to $\rho(\cdot) = \cN_d(0, I)$. For an initial state $z \in E$, the first random reflection time is distributed as  $\tau \sim \text{IPP}(t \mapsto \lambda_{\text{b}}(\phi(z, t)))$ with 
\[
\lambda_{\text{b}}(x,v) =   \max(\langle v, \nabla \Psi(x)\rangle s(x)   - \langle v,  \nabla s(x) \rangle , 0).
\]
At reflection time, the process changes velocity according to the  kernel $\cQ_{E, \text{b}}((x,v), \cdot) = \delta_{R_\Psi(x,v)}$ for
\[
R_\Psi(x,v) =  v - 2\frac{ \langle  v, \Psi(x) \rangle}{\|\Psi(x)\|}\Psi(x).
\]
At random exponential times with rate $\lambda_{\text{r}}(x,v) = c>0$, the process refreshes its velocity by drawing a new velocity $v' \sim \rho$.  See \textcite{2015arXiv151002451B} for an overview of the standard BPS. 
In this example, we extend the BPS for targets of the form 
\begin{equation}
    \label{eq: target bps teleport}
    \mu(\dd x) = \exp(-\Psi(x)) \ind_{A} \dd x
\end{equation}
for a set $A$ with a $(d-1)$-dimensional piecewise-smooth boundary $\partial A$ and a function $\Psi \in \cC^1(\RR^d)$. This corresponds to a smooth target density on a constrained space given by the set $A$.

While the standard BPS for constrained spaces as presented in \textcite{Bierkens_2018} would reflect the velocity every time the process reaches the boundary $\partial A$, in our setting, the kernel $\cT \colon  \partial \Omega\times \cB(\partial \Omega) \to [0,1]$, for some $ \partial \Omega \supseteq\partial A$ allows the process to jump in space when hitting $\partial \Omega$, effectively creating teleportation portals. To that end, we set 
$\cR_2(z,\cdot) = \delta_{R_n(z)}(\cdot)$ with 
\begin{equation}
    \label{eq: bps relfection at the boundary}
    R_n(x,v)  = v - 2\frac{ \langle  v, n(x) \rangle}{\|n(x)\|}n(x)
\end{equation}
while $\cR_1((x,v),y,\cdot) = \delta_{v'}(\cdot)$ with $v'(x,y,v)$  
such that 
$|v'| = |v|$ and $\langle v, n(x)\rangle = \langle v', -n(y)  \rangle$, for every $(x,v) \in \partial E^+, y \in \partial \Omega$. This corresponds to a rotation of the velocity vector which preserves the angles between the entrance and exit velocity vectors and normal vectors at the boundaries. 
The velocity of the sampler when bouncing or teleporting at the boundary can be randomized in the fashion of a generalized BPS (\cite{wu:hal-01968784}) by introducing a Crank-Nicolson update step in the orthogonal space relative to the normal vector of the boundary.

This framework allows the process to make jumps in $\partial \Omega$ and to visit disconnected regions or distant regions which are difficult to reach with continuous paths.
The problem of sampling from a conditional measure as in equation~\eqref{eq: target bps teleport} arises for example in the simulation of extreme events. In this case, standard Markov Chain Monte Carlo methods can fail to explore the full measure because of the inability of the chain to traverse subsets of measure close or equal to 0.

\section{Applications}\label{sec: applications}
In this section we motivate our work by applying the PDMPs with boundaries as described in Section~\ref{sec: pdmps samplers with boundaries} for sampling the latent space of infection times in the SIR model with notifications and for sampling the invariant measure of hard-sphere problems. 

\subsection{SIR model with notifications}\label{sec: epidemiology}
The model presented here is inspired by  the setting established in \textcite{jewell2009bayesian} for modelling  the spread of infectious disease in a population. Here we combine the PDMP for piecewise smooth densities with the framework presented in \textcite{bierkens2021sticky} for adding/removing efficiently in continuous time occult infected individuals (infected individuals which have not been notified up to the observation time) by means of introducing sticky events which are events after which the process sticks to lower dimensional hyper-planes for some random time.

Consider an infection process $\{Y(t) \in \{S, I,N,R\}^d \colon 0<t<T\}$ on a population of size $d$. Each coordinate $Y_i(t)$ takes values
\[
y_i(t) = \begin{cases}
    S & \text{if $i$ is \emph{susceptible} at time $t$},\\
    I & \text{if $i$ is \emph{infected} at time $t$},\\
    N & \text{if $i$ is \emph{notified} at time $t$},\\
    R & \text{if $i$ is \emph{removed} at time $t$}.\\
\end{cases}
\]
Each coordinate-process $(Y_i(t))_{t>0}$ is allowed to change state in the following direction:  $S\to I \to N \to R$.
 
For every individuals on $i = 1,2,\dots,d$, define 
\[
\tau_i =  \inf\{t>0 \colon y_i(t) = I\}, \quad \tau^\star_i =  \inf\{t>0 \colon y_i(t) = N\}, \quad \tau_i^\circ =  \inf\{t>0 \colon y_i(t) = R\}
\]
respectively for the first infection time, notification time, removing time of individual $i$ with the convention that $\tau_i^\star (\tau_i^\circ) = \infty$ if $\tau_i^\star (\tau_i^\circ)  \ge T$.  We observe $\tau^\star$ (the notification times) and $\tau^\circ$ (the removing times) and we are interested on recovering the minimum between the infection times and the observation time $T$: $ x := (x_i = \tau_i \wedge T \colon i =1,2,\dots,d)$. We this convention, $x_i = T$ when an individual $i$ has not been infected before observation time $T$, hence is susceptible at observation time.

For every $i=1,2,\dots, d$, individual $i$  changes its state from $S$ to $I$ according to an inhomogeneous Poisson process with rate $t \mapsto \beta_i(y(t))$ for a function $\beta_i \colon \{S,I,R,N\}^d \to \RR^+$ usually referred as the \emph{infectious pressure} on $i$. As  $\{y(t) \colon 0\le t \le T\}$ can be recovered by knowing  $(x,\tau^\star, \tau^\circ)$, we write $\beta_i(x) := \beta_i(y(x_i))$ (with this notation, we omit the dependence of $\beta_i(x)$ on $(\tau^\star, \tau^\circ)$ which are known and fixed throughout).
For each pair $(i,j)$ with $i\ne j$, define the \emph{infection rate}
\[
\beta_{i,j}(x) = \begin{cases}
 C_{i,j}, & x_i < x_j \le \tau_i^\star \\
 \gamma C_{i,j} &  \tau_i^\star < x_j \le \tau_i^\circ \\
 0 & \text{otherwise},
\end{cases}
\]
where $\gamma \in (0,1)$ is the factor of reduction of the infectivity after notification time and $C_{i,j} := d(i,j) \vartheta_{i}\xi_j$ is a measure of infectivity of $i$ towards $j$. Here  $d\colon \{1,2,\dots,d\}^2 \to \RR^+$ is an inverse distance metric between two individuals and  $\vartheta_i, \xi_i >0$ are seen  respectively as the infectivity and susceptibility baselines of individual $i$ (see \cite{jewell2009bayesian} for more details). The infectious pressure on individual $i$ is given by 
\[
\beta_j(x) = \sum_{i\ne j} \beta_{i,j}(x).
\]
Denote the set of individuals which have been notified before time $T$ by $\cN_{T^-} := \{i \colon \tau^\star_i < T\}$ and by $\cN_{T^-}^c := \{1,2,\dots,d\}\setminus \cN_T$ its complementary. We assume that the delay between infection and notification $\tau^\star_i - \tau_i \mid \tau_i$ of individual $i$ is a random variable with density $f$ and distribution $F$.

Then, the random vector $x = (x_1,x_2,\dots,x_d)$ is distributed according to 
\begin{equation} \label{eq: target sinr}
  \mu(\dd x) \propto \bigotimes_{i=1}^d \rho_i(x) \mu_i(\dd x_i).
\end{equation}
where $\rho_i(x) \mu_i(\dd x_i)$ can be heuristically interpreted as the distribution of the $i$th infection time. For $i\in \cN_{T^-}^c$,  
\begin{align} 
        \rho_i(x) &= (1 - F(T-x_i))\beta_i(x)\exp\left(-B_i(x)\right), \nonumber \\
      \mu_i(\dd x_i) &= \ind_{(0\le x_i \le T)} \dd x_i  + \kappa_i(x) \delta_T(\dd x_i)
      \label{eq: target measure not notified}
\end{align}
with $\kappa_i(x) = \frac{1}{\beta_i(x)}$ and $B_i(x) = \int_0^T \beta_i(x[i; s]) \dd s$. See Appendix~\ref{app: derivation of the measure in epidemiology} for the details of the derivation of the measure above. Here the point mass at $x_i = T$ absorbs the event that individual $i$ has not been infected before time $T$ while the remaining part of the density represents the event for the individual $i$ to be infected but not notified before time $T$ (in such case the infected individual $i$ is often referred as \emph{occult}). 
For $i\in \cN_{T^-}$,  
\begin{align} 
 \rho_i(x) &= \beta_i(x)\exp\left(-B_i(x)\right)f(\tau^\star_i - x_i), \nonumber \\
 \mu_i(\dd x_i) &= \ind_{(0\le x_i \le \tau^\star_i)} \dd x_i, \label{eq: target measure notified}
\end{align}
see Appendix~\ref{app: derivation of the measure in epidemiology} for details.

We apply the Zig-Zag sampler for discontinuous densities as presented in Section~\ref{sec: Zig-Zag sampler for discontinuous densities}, with no speed-up function $(s(x) = 1)$ and with sticky events to target the measure $\mu$. Below, we summarize the behaviour of the process:  
\begin{itemize}
    \item The process reflects velocity randomly in space according to the gradient of the continuous component of the density of $\mu$. The reader may find in  Appendix~\ref{app: Computing reflections times and probabilities to cross discontinuities} the explicit computations of the reflection times. 
    \item For every $i = 1,2,\dots,d$, the process hits the boundary when the coordinate  $x_{i}$  hits a element of the vectors $x_{-i}, \tau^\star, \tau^\circ$. This is because $x_i \mapsto \beta_{i}(x)$ is discontinuous on those points. For the boundary corresponding to two coordinates of $x$ colliding, the process bounces off the discontinuity with some probability by changing the sign of the velocity of both coordinates, while if the discontinuity corresponds to a coordinate of $x$ colliding with a notification/removal time, the process bounces off just by changing the sign of the velocity of that coordinate.
    Upon hitting a given boundary, the process traverses the discontinuity without changing its velocity with some probability.
    In particular, each coordinate-process $(x_i,v_i)$ for $i=1,2,\dots,d$ never crosses the boundary $(x_i, v_i) = (\tau_i^\star, +1)$ and $(x_i, v_i) = (0, -1)$. 
    \item Each particle $(x_i, v_i)$ with $i \in \cN^c_T$ sticks at $T$ for an exponential time with rate equal to $\beta_i(x)/2$ (the factor $\frac{1}{2}$ originates since the Dirac measure is located at the boundary of the interval, see remark~\ref{rmk: sticking at the boundary}), upon hitting $T$. The sticking time of the particle $x_i$ corresponds to the individual indexed by $i$ being susceptible and not infected.
\end{itemize}

 \subsubsection{Numerical experiment}
 We fix $d = 50$ and set $\xi_1, \xi_2,\dots,\xi_d,\vartheta_1,\dots,\vartheta_d$ to be the realization of i.i.d. random variables distributed according to $\text{Unif}([0,1])0.9 + 0.7$ and set $d(i,j) = 0.4 (\ind_{(|i - j| \le 5)})$. We assume that $\tau^\star_{i} - \tau_{i} \mid \tau_i \sim \text{Exp}(0.3)$. 
 
 To generate a synthetic dataset, we simulated forward the model up to time $T=5.0$, setting $\tau_{25} = 0$, see Figure~\ref{fig:foward_sim} for visualizing the dynamics of each individual. Before time $T$, 47 individuals have been infected, 28 individual have been notified, 18 individuals have been removed.

 \begin{figure}[!ht]
     \centering
     \includegraphics[width = 1.0\textwidth]{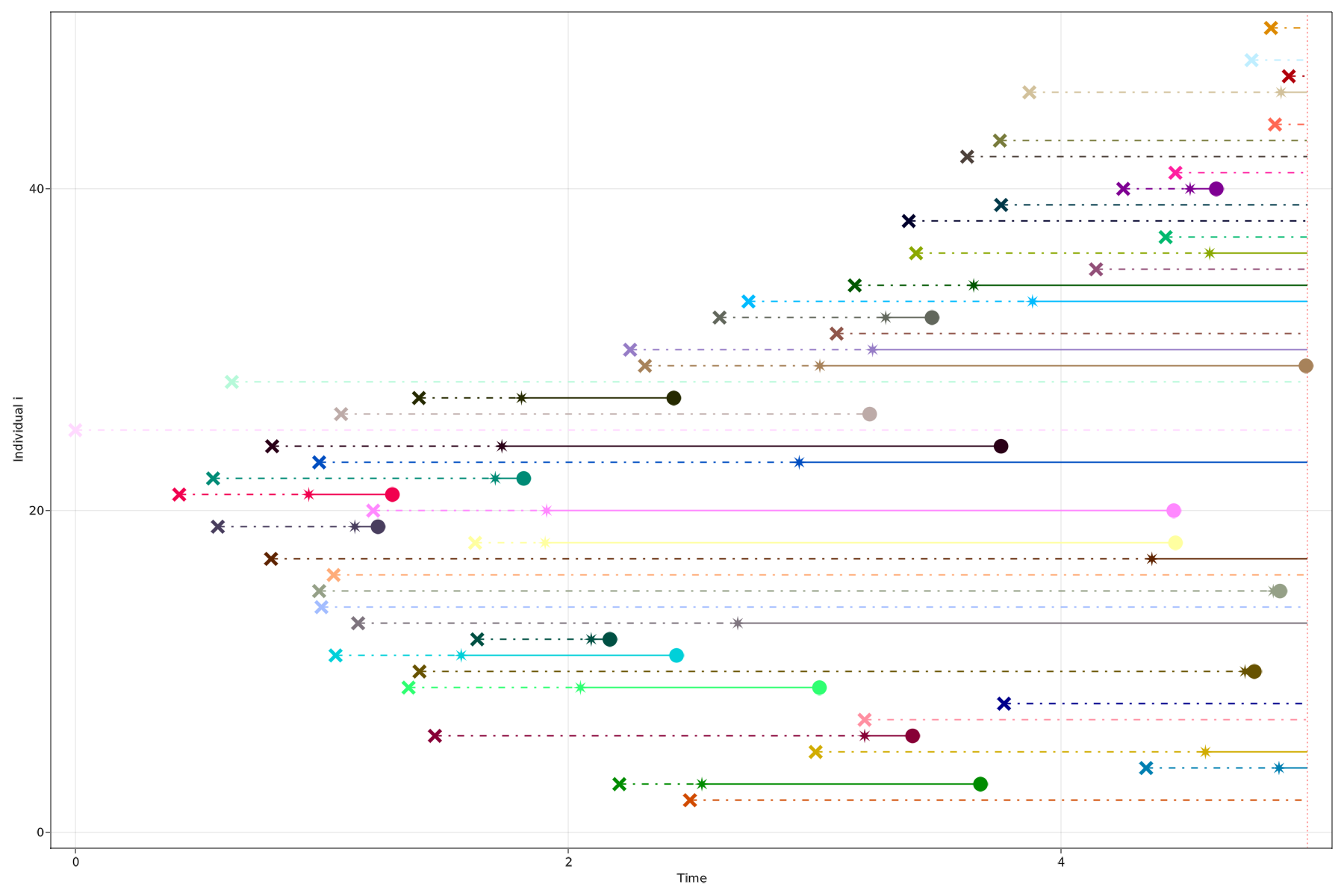}
     \caption{Simulated data from the  SINR model of the population $y$-axis (conditional on the infection time of individual 25 to be $0$). The symbols $\times, \star, \circ$ indicates respectively the unobserved infection times, the observed notification times and the observed removal times for each individual. The dotted lines indicate the unobserved time between infection and notification times.
     \label{fig:foward_sim}}
 \end{figure}

 We fix $x_{25} = \tau_{25} = 0$ and simulate the $(d-1)$-dimensional sticky Zig-Zag sampler with final clock $T^\star = 500$. Figure~\ref{fig:some_traces} shows the marginal posterior densities and the final segment of the Zig-Zag trajectory relative to coordinates 11, 13, 38, 49. Those individuals  have a different status at time $T$: susceptible,  occult  (infected but not notified), notified and removed. 
 
  \begin{figure}[h!]
     \centering
     \includegraphics[width = 0.9\textwidth]{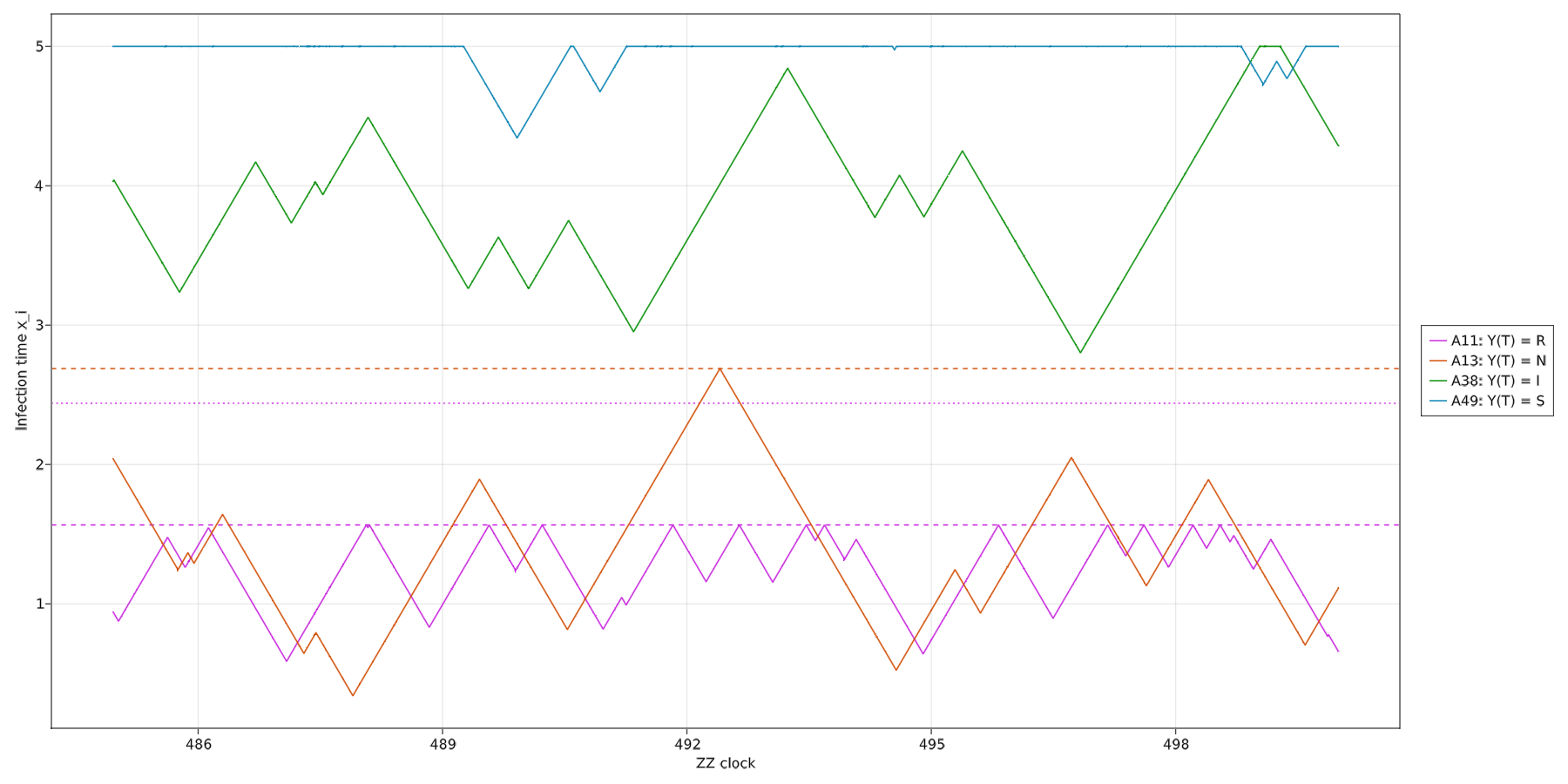}
     \includegraphics[width = 0.9\textwidth]{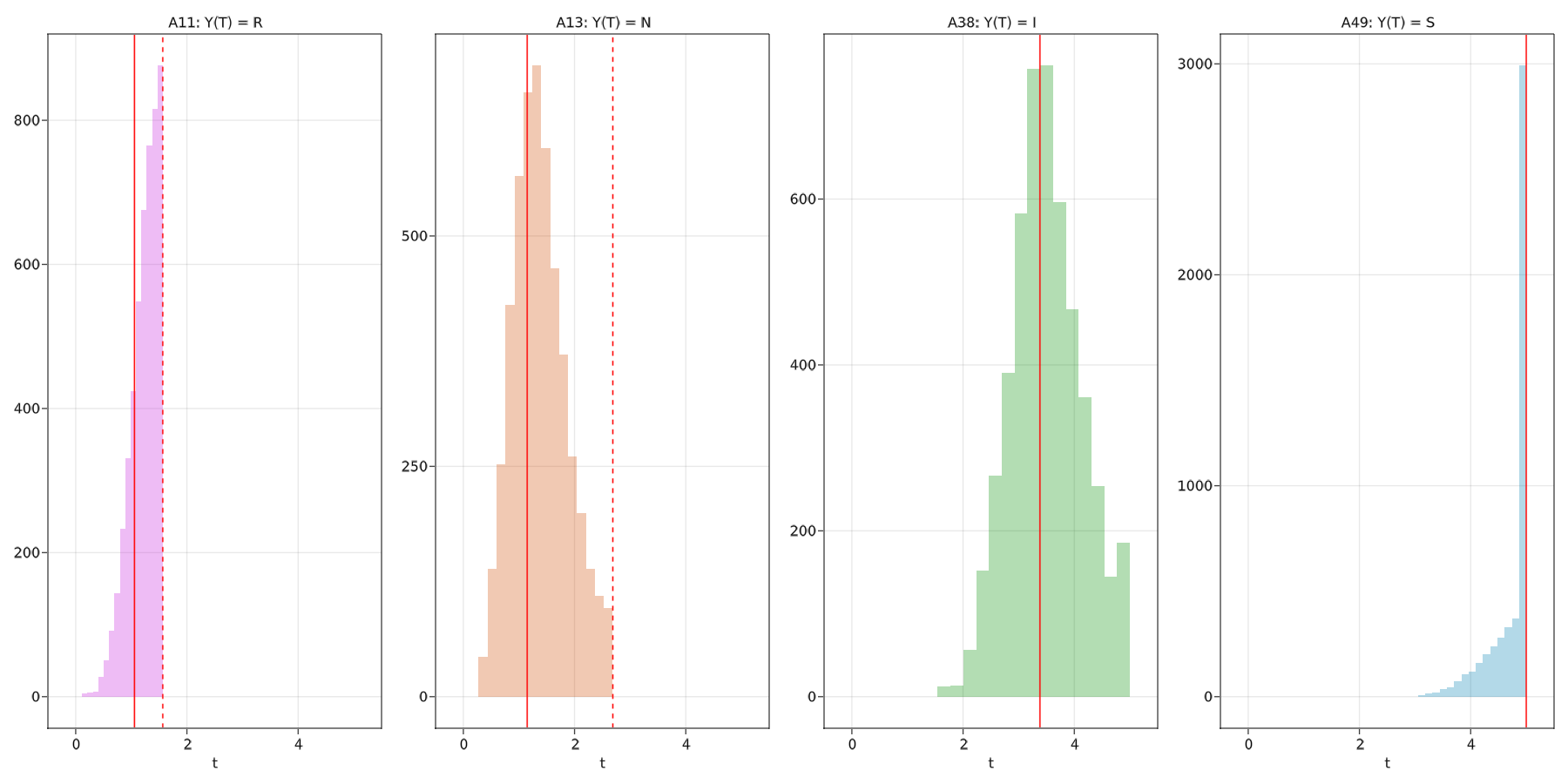}
     \caption{Top panel: Final 25 time units of the coordinate-process relative to the infection times of individuals 11  (pink), 13 (orange), 38 (green), 49 (blue) which, at observation time $T$ had a different status (susceptible, infected but not notified, notified  and removed). Dashed lines corresponds to the notification times, dotted lines corresponds to the removing time. Bottom panel: marginal densities of the infection times of those individuals estimated with the full trajectory of the Zig-Zag.
     }
     \label{fig:some_traces}
 \end{figure}

\subsection{Hard-sphere models} \label{sec: teleportation}
     As an application for teleportation which is relevant in statistical mechanics, we consider the sampling problem in hard-sphere models. This class of models motivated the first Markov chain Monte Carlo method in the pioneering work of \textcite{metropolis1953equation}. More recently, PDMP samplers have been employed for this class of problems  for example in \textcite{Michel_2019} and \textcite{monemvassitis2022pdmp}. For an overview of hard-sphere models, see \textcite[Chapter 2]{krauth2006statistical}.  For a survey of MCMC methods used for sampling from hard-sphere models, see \textcite{faulkner2022sampling}. For related models, see \textcite{moller2010perfect}. In this section, we use the framework introduced in Section~\ref{sec: pdmps samplers with boundaries} to define teleportation portals which swap pairs of hard-spheres when they collide. This idea has been already investigated (not in the context of PDMPs) for constrained uniform measures and referred as \emph{Swap Monte Carlo move}, see for example \textcite{berthier2016equilibrium}. 
     
     We consider $N$ particles, each one taking values in $\RR^d$. Denote the configuration of all particles by $x = \{x^{(i)} \in \RR^d \colon 1\le i \le N\}$ where we identify the location of the $i$th particles by $x^{(i)} = x_{[(i-1)d + 1, id]}$. Consider a measure $\mu^\star (\dd x) = \exp(-\Psi(x)) \dd x$, where $\Psi(x) = \sum_{i=1}^N  \Psi_0(x^{(i)})$, for a smooth function $\Psi_0$ supported on $\RR^{d}$. We assume that each particle $i=1,2,\dots,d$ is a hard-sphere centered in $x^{(i)}$ and with radius $r_i>0$  and consider the conditional invariant measure
    \[
    \mu(\dd x) \propto \mu^\star(\dd x) \ind_{x\in A}
    \]
    with  $A = \bigcap_{i =1}^N \bigcap_{\substack{j = 1,\\j \ne i}}^N A_{i,j}$ and 
    \[
    A_{i,j} = \{ x \in \RR^{dN} \colon \|x^{(i)} - x^{(j)}\| \ge (r_i +r_j)\},
    \]
    that is the measure $\mu^\star$ conditioned on the space where all the hard-spheres do not overlap.
    The restriction for the process to in $A$ creates boundaries which slow down the exploration of the state space for standard PDMP samplers as  they must reflect the velocity vector when hitting the set $A^c$, thus never crossing the boundary, see \textcite{Bierkens_2018} for a detailed description of standard PDMPs on restricted domains. 
    
    We apply the $(Nd)$-dimensional Bouncy Particle Sampler described in Section~\ref{sec:  Bouncy Particle sampler with teleportation} with no speed-up ($s(x) = 1$). To that end, we define the boundary of the process 
    \[\partial E^+ = \{ (x,v) \in \partial \Omega\times \RR^{dN} \colon  \langle n(x), v\rangle > 0 \}\] 
    where $\partial \Omega  = \bigcup_{i =1}^N \bigcup_{\substack{j = 1\\j \ne i}}^N \partial \Omega_{i,j}$ and
    \[
    \partial \Omega_{i,j} = \{ x \in \RR^{dN} \colon |x^{(i)} - x^{(j)}| = (r_i +r_j)\}.
    \]
   We define the kernel $\cT(x, \cdot) = \delta_{\kappa(x)}(\cdot)$ for a function $\kappa\colon \partial \Omega \to \partial \Omega$. For $x \in \partial \Omega_{i,j}$, we set 
   \[
[\kappa(x)]^{(\ell)} = \begin{cases}
         x^{(j)} + \frac{x^{(i)} - x^{(j)}}{r_i + r_j}(r_i - r_j) &\text{ if } \ell = i\\
    x^{(i)} + \frac{x^{(j)} - x^{(i)}}{r_i+ r_j}(r_j - r_i)  &\text{ if } \ell = j \\
    x^{(\ell)} &\text{ otherwise }\\
\end{cases} 
\]
 which attempt to swap the location of the balls $i$ and $j$ by moving the smaller hard-sphere more than the larger hard-sphere, while preserving the location of the extremities of the two hard-spheres. The teleportation is successful with non-zero probability only if  $\kappa(x) \in A$,  that is, if after teleportation, no hard-spheres overlaps, see Figure~\ref{fig: teleportation 1} for an illustration. Other possible choices of teleportation portals are possible, see Appendix~\ref{app: hard-sphere model} for a discussion.

The Bouncy Particle Sampler with teleportation at the boundary behaves as follows. For any point $(x,v) \in \partial E^+$ and for $\ell = 1,2,\dots,N$:
    \begin{itemize}
    \item propose to teleport to $y = \kappa(x)$; 
    \item if $\kappa(x)\in A$, then with probability $\alpha(x, \kappa(x))$, set the new state equal to $(\kappa(x), w)$ with $w =  R_{n}(\kappa(x), -v)$; 
    \item otherwise reflect the velocity at the boundary and set the new state equal to $(x, w)$ with $w = R_{n}(x,v)$,
\end{itemize}
where $R_n$ is defined in equation~\eqref{eq: bps relfection at the boundary}.
\begin{figure}[h!]
    \centering
    \includegraphics[width = 0.6\textwidth]{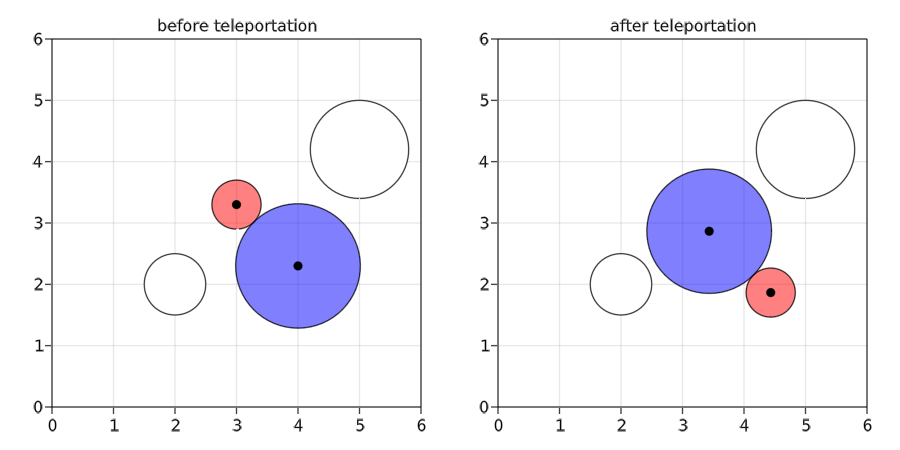}
    \includegraphics[width = 0.6\textwidth]{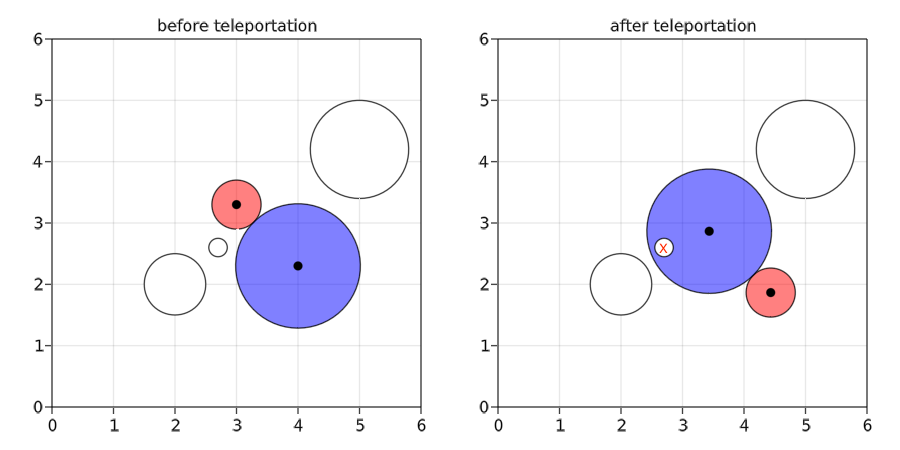}
    \caption{Illustration of the teleportation portals. Left panels: two configurations $x \colon (x,v) \in \partial E^+$. Right panels: $y = \tau(x)$.
     Top panels: the proposed state after teleportation is valid, i.e. $\tau(x) \in  A$. Bottom panels: the proposed state after teleportation is not valid i.e. $\tau(x) \in A^c$.}
    \label{fig: teleportation 1}
\end{figure}

\subsubsection*{Numerical experiment}
We fix $N = 6$ and $d = 2$. We let $r_i \sim 2.0 + 1.5\text{Unif}([0,1])$. We set 
\[\Psi_0(x) = \frac{1}{4}\|x\|^2, \quad x \in \RR^d\] 
 and compare the performance of the standard BPS on constrained space as presented in \textcite{Bierkens_2018} and the BPS with teleportation as presented in Section~\ref{sec:  Bouncy Particle sampler with teleportation}. We initialize both the samplers in a valid configuration, see Figure~\ref{fig: initialization of bps}. Both samplers have refreshment rate $\lambda_{\mathrm{r}, E}(z) = 0.01$.
 
 Figure~\ref{fig: angles} compares the trace of the functional $\langle x_i, x_j \rangle$ where $i,j \in \{1,2,\dots,N\}$ are the indices of the hard-spheres with largest radius: $r_i = \max(r), r_j = \max(r_{-i})$ obtained when running the  two samplers with final clock equal to 2,000. The animations showing the evolution of the hard-spheres according to the dynamics of the standard BPS and BPS with teleportation may be found at \sloppy{\texttt{https://github.com/SebaGraz/hard-sphere-model}}.

\begin{figure}[h!]
    \centering
    \includegraphics[width = 0.6\textwidth]{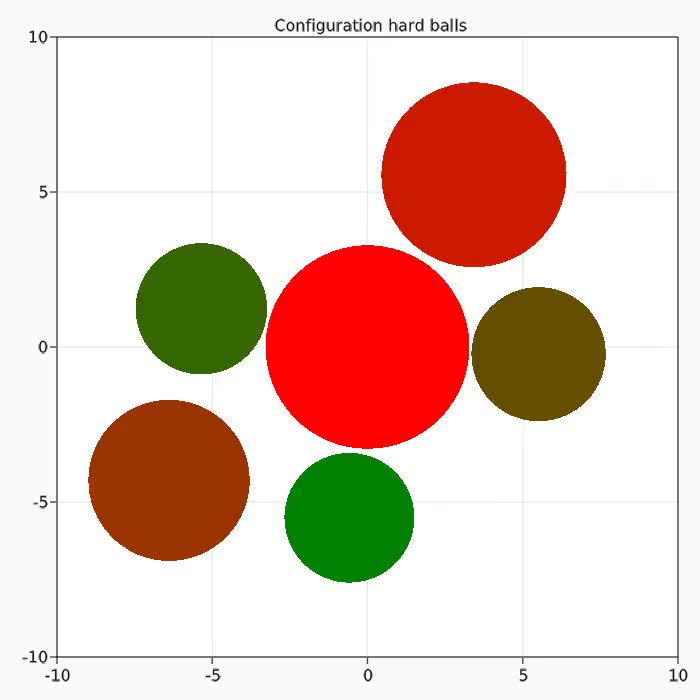}
    \caption{Initial configuration $x(0)$ for the standard BPS and BPS with teleportation. Here $x(0)\in A$.}
    \label{fig: initialization of bps}
\end{figure}

\begin{figure}[h!]
    \centering
    \includegraphics[width = 0.49\textwidth]{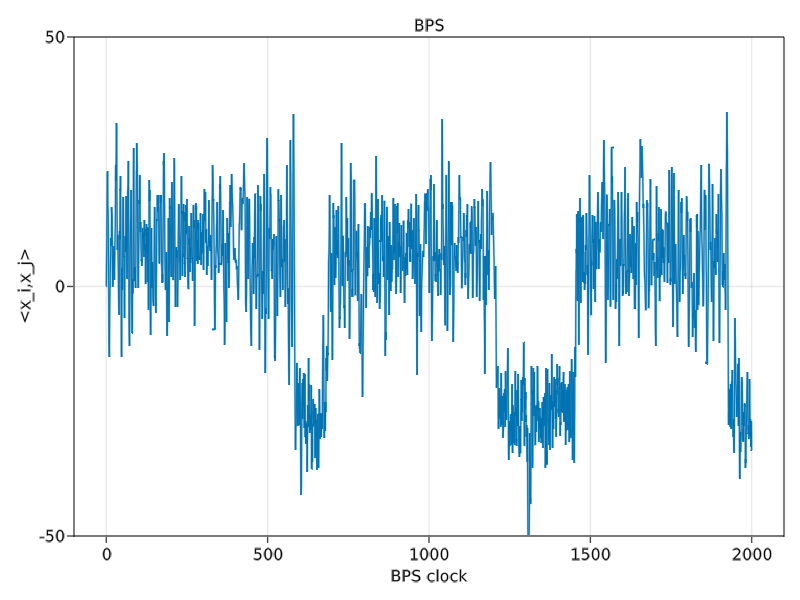}
    \includegraphics[width = 0.49\textwidth]{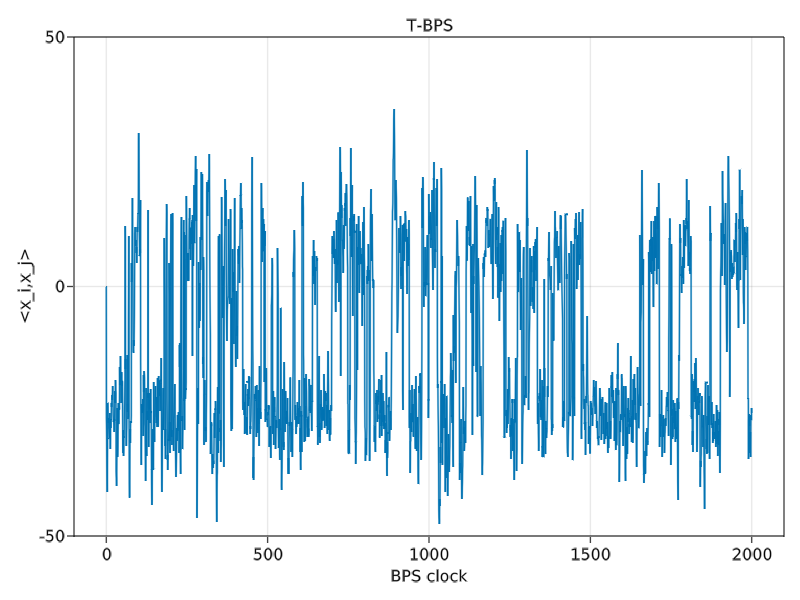}
    \caption{trace $t \mapsto \langle x_i(t), x_j(t) \rangle$ where $t$ is the clock of the Bouncy Particle Sampler (BPS) and $i,j$ are the indices of the two hard-sphere with largest radius. Left: BPS without teleportation.  Right: BPS with teleportation.}
    \label{fig: angles}
\end{figure}

\section{Discussion}\label{sec: discussion}
The material presented in this work suggests the following research directions.

Firstly, the theoretical framework presented in Section~\ref{sec: pdmps samplers with boundaries} introduces a speed-up function $s$ which, as discussed in Remark~\ref{rmk: sppedup} and Remark~\ref{rmk: sppedup 2} can be chosen to be constant in each set $\Omega_i,\, i \in K$ and can be tuned to increase the probability to cross the boundary (provided that Assumption~\ref{ass: speed growth condition} is satisfied) or to approximate with piecewise constant functions a continuous speed-up function for which the deterministic dynamics cannot be integrated analytically. In Section~\ref{sec: applications}, we set $s(x) = 1$ and we have not yet investigated the benefits of tuning the speed-up function. It was proven in \textcite{vasdekis2021speed} that the speed-up function is beneficial for heavy tailed distribution and we expect it can improve the performance of the sampler for multi-modal densities.

Secondly, in Section~\ref{sec: epidemiology}, we use PDMPs for sampling the latent space of infection times in the SINR model. In contrast with the method proposed for example in \cite{jewell2009bayesian}, this framework allows to set the state of each individual in $\cN^c_T$ from susceptible to infected (and vice-versa) continuously in time, without  using the framework of Reversible-Jump MCMC and without defining tailored proposal kernel for the trans-dimensional moves when changing individuals from susceptible to infected and vice-versa. Furthermore the Zig-Zag sampler, due to  1) its non-reversible nature, 2) its ability of sampling the minimum between the observation time and the infection time and 3) its `local' implementation which exploits the sparse conditional independence structure of the target measure to gain computational efficiency (see e.g. \cite[Section~4]{bierkens2020piecewise}), shows promise to scale well with respect to the population size.  
The analysis of mixing times of the sampler and the scaling of the complexity of the algorithm remains an open challenge.

Thirdly, in Appendix~\ref{app: hard-sphere model}, we discuss  some possible choices of teleportation portals for hard-sphere models which can improve the mixing times of the sampler. The list of teleportation portals considered is far from exhaustive and a systematic study on the choice of teleportation portals for hard-sphere models is left for future research. 

Finally, the framework presented in this work can be leveraged to define boundaries, including teleportation portals, also when targeting differentiable densities with respect to the Lebesgue measure. This allows the process  to jump in space and improve its mixing properties. This research direction is not explored here.  

\printbibliography
\newpage
\appendix

\section{Theoretical results}\label{app: theoretical results}
\subsection[Extended generator]{Extended generator of PDMPs with boundaries}\label{app: generator of PDMPs}
The extended generator for PDMPs with boundaries is given in \textcite[Section 26]{Davis1993} and, for the PDMP samplers considered in Section~\ref{sec: pdmps samplers with boundaries}, takes the form 
\[
\begin{split}
\cL f(x,v) = \langle v, \nabla f(x)\rangle s(x) + \lambda_{ \mathrm{b}}(x,v) \int \cQ_{E, \mathrm{b}}((x,v), \dd z) (f(z) - f(x,v)) \\ + \lambda_{\mathrm{r}}(x,v) \int \cQ_{E,\mathrm{r}}((x,v), \dd z) (f(z) - f(x,v)).
\end{split}
\]
Here we let $\cL$ act on functions $f$ in the set
\begin{align*}
    \cA = \{f \in C_c(E); \quad & \, t\mapsto f(\phi(z,t)) \text{ is absolutely continuous } \forall z \in E;  \\
    &   f(z)= \int_{\partial E^-} f(z') Q_{\partial E}(z, \dd z'), \quad \forall z \in \partial E^+\}.
\end{align*}
The set $\cA$ is contained in the domain of the extended generator $\cD(\cL)$ given in \textcite[Section 26]{Davis1993}.
\subsection{Proof Proposition~\ref{prop: integral of generator with standard assumptions}}\label{app: proof proposition condiotions on the interior}
By applying the divergence theorem, we have that, for $f \in \cA$,
\begin{align}
    \int_E \cL f \dd (\mu \otimes \rho) &= \int_{E} \left( \langle v,  s(x) \nabla \Psi(x) - \nabla s(x) \rangle - \lambda_{\mathrm{b}}(x,v) \right) f(x,v) \rho(\dd v) \mu(\dd x)  \label{eq: integral of generator 1}\\
    &\quad + \int_{E} \int_E \lambda_{\mathrm{b}}(z) \cQ_{E,\mathrm{b}}(z, \dd z') f(z') (\rho\otimes\mu)(\dd z)\label{eq: integral of generator 2}\\
    &\quad  +\int_{E} \lambda_{\mathrm{r}}(z)\int_E \cQ_{E,\mathrm{r}}(z, \dd z') (f(z') - f(z)) (\rho\otimes\mu)(\dd z) \label{eq: integral of generator 3}\\
    &\quad + \int_{\partial E} \langle v, n(x)\rangle f(x,v)s(x) \rho(\dd v) \partial \lambda(\dd x) \label{eq: integral of generator 4}.
\end{align}

Then, by Assumption~\ref{ass: conditions on the interior}, the right hand-side of \eqref{eq: integral of generator 1} cancel with the term in \eqref{eq: integral of generator 2}. Furthermore, by Assumption~\ref{ass: conditions on the refreshments}, the refreshment kernel is invariant to $\rho\otimes\mu$ so  the term in~\eqref{eq: integral of generator 3} is also equal to 0. We are left with the term in~\eqref{eq: integral of generator 4},  proving Proposition~\ref{prop: integral of generator with standard assumptions}.
 \subsection{Open set irreducible PDMPs}\label{app: open set irreducible}

 \begin{lemma}\label{prop: connected to path} Let $A \subset \RR^d$. If $A$ is an open and connected set, then for any two points $x,y \in A$, there exists a continuous and piecewise linear path in $A$ connecting $x$ and $y$.
\end{lemma}
\begin{proof}
    Denote the set $P := \{z \in A \colon \text{there exists a continuous and piecewise linear path connecting } x \text{ and } z\}$. In the following, we show that $P$ is both open and closed in the relative topology and is not empty. Since $A$ is connected, the only two sets which are both open and closed must be $A$ and $\emptyset$. Hence we conclude that $P = A$.

    ($P$ is not empty). Since $A$ is open, there exists an $r > 0$ with $B_r(x) \subset A$. Since $B_r(x)$ is convex, we can connect $x$ with any point $y \in B_r(x)$ with a linear segment, so that $P$ is not empty.

    ($P$ is open). For any $z \in P$, every  $y \in B_r(z)$ is path connected to $z$ with a linear segment, hence it is connected to $x$ as well by piecewise linear paths. We conclude that $B_r(z) \subset P$. Therefore $P$ is open.

    ($A\setminus P$ is open. Therefore $P$ is closed).  For every $x \in A \setminus P$,  $\exists \, r>0, \, B_r(x) \subset A$. Furthermore $B_r(x) \cap P = \emptyset$, since otherwise we could connect $x$ with a linear segment to the region $P$. Hence $B_r(x) \subset A\setminus P$. Therefore $A\setminus P$ is open and $P$ is closed. 
\end{proof}
 We present now a straightforward corollary of Lemma~\ref{prop: connected to path}. Below, the term `admissible path' was introduced in Definition~\ref{def: admissible paths} in the manuscript. Recall that $E_i = E_i^\circ \cap \partial E_i^-$, where $ E_i^\circ$ is open and $\partial E_i^-$ contains those points on the boundary obtained by applying the flow $\phi(z, \cdot)$ for some $z \in E^\circ_i$ backward in time.

\begin{corollary}\label{cor: exists a path connecting points} 
For every $i \in K$ and every two points $z, \,z' \in E_i^\circ$, there exists an admissible path $(z(t))_{0\le t\le T}$ with $z(0) = z$ and $Z_T = z'$. 
\end{corollary}

In the following, we fix $z, z'$ and such admissible path $(z(t))_{0\le t\le T}$ and let $(z^{(j)} = \{x^{(j)}, v^{(j)})\}_{j=0,1,2,\dots,N}$ with $z^{(0)} = z$ and $z^{(N)} = z'$ be the finite collection of points where the path changes its velocity component.

As $E^\circ_i$ is an open set in the product topology of $\Omega \times \cV$, for every $t\in [0,T]$, there exists a radius $r(t)>0$, with $B_{r(t)}(x(t)) \subset \Omega_i$. Let $r_0 = \min_{0\le t\le T}\{r(t)\} > 0$ and define a \emph{tubular neighbourhood} of the admissible path $(z(t))_{0\le t\le T}$ as
\begin{equation}
    \label{eq: app tubular neigh}
    B := \{B_{r_0}(x(t)), \,t\in [0,T] 
    \} \subset \Omega_i.
\end{equation}


For $j=0,1,2,\dots,N$, define the set 
\[
V_j(x) := \begin{cases} 
\{v^{(j)}\} & j=0\\
\{v \in \cV \colon \|v\|\le 1 \text{ and }\, \exists t, \, \phi_x(t, (x, v)) \in B_{r_0}(x^{(j+1))}\} & 1\le j < N \\
\{v \in \cV \colon (x,v) \in B_{r_0}(z^{(j)})\} & j = N
\end{cases}
\]
where $B_r(x,v) := B_r(x)\times B_r(v)$ is the open ball with radius $r$ relative to the product topology of position space and velocity space of the process. Notice that if $\cV$ is a finite set, $B_r(x,v)$ can contain only a singleton $v \in \cV$.  For $j=0,1,\dots,N-1$, define the sets
\[
\Delta_j(z) := \{t \colon \phi_x(t, z) \in B_{r_0}(x^{(j+1)})\},
\]
and 
\[
U_j := \{(y,w) \in E_i \colon y\in B_{r_0}(x^{(j)}), w\in V_j(y)\}
\]
We now present some preliminary results which will be used in Proposition~\ref{prop: probability to move through skeleton points}:
\begin{remark}\label{rmk: non empty set}
    For every $j = 0, 1,\dots, N,$ the set $U_j$ is non-empty. Indeed for every $x \in B_r(x^{(j)})$, we have that $V_j(x)$ is not empty as it contains at least the point $v^{(j)}/\|v^{(j)}\|$.  
 \end{remark}  
 \begin{remark}\label{rmk: set with volume}
    For every $(x,v) \in U_j$, $|\Delta_j| > 0$ since $\{t \ge 0 \colon x + vt \in  B_r(x^{(i+1)})$ is a non-empty open 1-dimensional interval.
\end{remark}
\begin{remark}\label{rmk: inside tubolar neighbourhood}
    For every $(x,v) \in U_j$, the continuous trajectory $t \mapsto \phi((x,v), t)$ does hit a boundary before $\max(\Delta_i(x,v))$. In fact, for $t \le \max(\Delta_i(x,v))$, $\phi_x(t, (x,v)) \subset B \subset \Omega_i$.
\end{remark}
 \begin{lemma}\label{rmk: positive volume of velocity bps}\emph{(Positive probability of the BPS to choose the right direction when refreshing)}
    Let $\cV = \RR^d$ and assume that $\rho(\dd v) = \varrho(v) \dd v$ for a continuous density $\varrho$ supported in $\cV$. For every $j = 1,2,\dots,N$ and for every $x \in B_{r_0}(x^{(j)})$, $\rho(V_j(x)) > 0$. 
\end{lemma}
\begin{proof}
    It is clearly enough to show that $V_j(x) \subset \RR^d$  has positive $d$ dimensional Lebesgue measure.  This is trivial for $j = N$ (see definition of $V_j(x)$). Now fix $j \in\{ 1,\dots, N-1\}$ and $(x,v) \in U_j$ (By Remark~\ref{rmk: non empty set}, there is at least one element in this set). Let $y = \phi_x((x,v), t) \in B_{r_0}(x^{(i+1)})$ for some $t>0$. Then there exists a $r_1$, for which $B_{r_1}(y) \subset B_{r_0}(x^{(i+1)})$. For every $x' \in B_{r_1}(y)$, there exists $t' \ge 0$ and $v' \in \mathcal V$ such that $\phi_x((x,v'), t') = x'$ and $\phi_x((x,v'),s) \subset B$ for $0\le s \le t'$. Hence $c v'/\|v'\|  \in V_{j}(x), \, 0 \le c \le 1$ and $V_{j}(x)$ is a open and non-empty set. Therefore $|V_{j}(x)|>0$.  
\end{proof}
For a fixed point $z \in U_j$, denote the random variables $\tau^{(i),\rr} \sim \text{IPP}(s \mapsto \lambda_{\rr}(\phi(z,s)))$ and  $\tau^{(i),\bb} \sim \text{IPP}(s \mapsto \lambda_\bb(\phi(z,s)))$, which correspond to the proposed refreshment and reflection times at iteration $i$ of Algorithm~\ref{alg: algorithm abstract} and, conditional on $\tau^{(i),\rr}$  and $\tau^{(i),\bb}$, the random variable $w^{(i+1)} \sim \cQ_{\rr}(\phi(z, \min(\tau^{(i),\rr},\tau^{(i),\bb})),  \cdot)$ the random variable given by the refreshment kernel of the velocity acting on $z$ as in Assumption~\ref{ass: conditions on the refreshments}. 

For a fixed point $z \in E$, recall that we can simulate the first random event of a PDMP with Poisson rate $\lambda(z) = \lambda_\rr(z) + \lambda_\bb(z)$ and Markov kernel $\cQ(z, \cdot) = \frac{\lambda_\bb(z)}{\lambda(z)} \cQ_\bb(z, \cdot) + \frac{\lambda_\rr(z)}{\lambda(z)} \cQ_\rr(z, \cdot)$ by superimposing the two events (see Section~\ref{subsec: conditions on the interior}).

\begin{proposition}\label{prop: probability to move through skeleton points} \emph{(Probability to move between any two skeleton points)}
Suppose that Assumptions~\ref{ass: conditions on the refreshments}-\ref{ass: conditions on the interior} hold. For any $i=0,1,2,\dots,N-1$ and any point $z\in U_i$, define the event 
\[
E_i = \{ (\tau^{(i),\rr}, \tau^{(i),\bb}, w^{(i+1)}) \in \RR^+\times\RR^+\times\cV \colon \tau^{(i),\rr} \in \Delta_i,\,  \tau^{(i),\bb} > \tau^{(i),\rr},\, w^{(i+1)} \in V_{i+1}(\phi_x(\tau^{(i),\rr}, z))\},
\]
we have that $p(z) = \PP_z(E_i)> 0$, for the Bouncy Particle Sampler as well as Zig-Zag Sampler with velocity $\{-\frac{1}{\sqrt{d}}, \frac{1}{\sqrt{d}}\}^d$.
\end{proposition}
\begin{proof}
     By Remark~\ref{rmk: set with volume} $|\Delta_i|>0$. Then  
    $
    P(\tau^{(i), \rr} \in \Delta_i) 
    > 0
    $
    as, by Assumption~\ref{ass: conditions on the refreshments}, $0<\lambda_\rr(z) \le M, \forall z \in E$. Furthermore, by Remark~\ref{rmk: inside tubolar neighbourhood} we have that $\phi((x,v), s) \in E$, $0\le s\le \max(\Delta_i)$. Next,
    \[
    \PP(\tau^{(i), \bb}>\tau^{(i),\rr}) > \PP(\tau^b>\max(\Delta_i)) = \exp(-\int_0^{\max(\Delta_i)} \lambda_\bb(\phi((x,v), s)) \dd s) 
    \]
    which is strictly greater than 0 as $s \mapsto \phi(s, x,v)$ is continuous everywhere in $E_i$, and $\lambda_\bb$ as in Assumption~\ref{ass: conditions on the interior} is therefore bounded in the interval $[0, \max(\Delta_i)]$. Finally, for every $\tau^{(i), \rr}\in \Delta_i$ we have that $y = \phi_x((x,v), \tau^{(i),\rr}) \in B_r(x^{(i+1)})$ and, as noticed in Remark~\ref{rmk: non empty set}, $V^{i+1}(y)$ is non-empty. By Assumption~\ref{ass: conditions on the refreshments}, if $\cV$ is a finite set (as in the case of the Zig-Zag sampler), $\PP(w(y) = v^{(i+1)})  >0 $, if $\cV = \RR^d$ as in the Bouncy Particle sampler, $\PP(w(y) \in V_{i+1}(y)) >0$ as noted in Lemma~\ref{rmk: positive volume of velocity bps}.
\end{proof}

Let $\tau_A = \{\inf_{t\ge0} Z(t) \in A\}$ and for $j=0,1,2,\dots,N-1$.
\begin{corollary}\label{cor: reachability}
    Suppose that Assumptions~\ref{ass: conditions on the refreshments}-\ref{ass: conditions on the interior} hold. Then for $i \in K$, $z \in E_i$ and $A \subset E^\circ_i$ open and non-empty, there exists a $t>0$ such that 
    \[
    \PP_t(z, A)>0
    .\]
\end{corollary}

\begin{proof}
     Take an element $z' \in A$ and construct an admissible path connecting $z$ with $z'$, changing direction at the finite collection of points, $(z^{(i)})_{i=1,2,\dots,N}$. We can then use Proposition~\ref{prop: probability to move through skeleton points} iteratively and conclude that 
    \[
    \PP\left( \big(Z(t_1) \in U_1 \mid Z(0) = z \big) \bigcap_{j=2}^N (Z(t_j) \in U_j \mid Z(t_{j-1}) \in U_{j-1})\right)>0,
    \]
    for a sequence $0 = t_0 < t_1 <t_2<\dots<t_N<\infty$.
\end{proof}
The generalization of Proposition~\ref{cor: reachability} for every initial point $z \in E$ and open set $A \subset E$ is straightforward:
\begin{proposition}\label{app:  prop reachability 2}
Suppose Assumptions~\ref{ass: conditions on the refreshments}-\ref{ass: conditions on the interior}-\ref{ass: sdb at the boundary}-\ref{ass: connectivity} hold. Then, for every $z \in E$, open set and non-empty set $A \subset E$, and for every $\Delta>0$, there exist an integer $t>0$ such that
\[
\PP_{t} (z, A) >0. 
\]
\end{proposition}
\begin{proof}
By Assumption~\ref{ass: connectivity}, for every $z \in E$, $A \subset E$, there is a sequence of sets $E_{k_1} \curvearrowright E_{k_2}\dots\curvearrowright E_{k_n}$, with $z \in E_{k_1}$, $A \subset E_{k_n}$ and sets $A_i\subset \partial E_{k_i}^+$ and $B_{{i+1}} \in \cB(\partial E_{k_{i+1}}^-)$  as in Definition~\ref{def: direct neighbour} for $i=1,2,\dots,n-1$. Denote the sets $B_{0} = \{z\}, \, A_{n} = A$ and take any points $z^{i}_B \in B_{i}, z^{i+1}_A \in A_{{i+1}}$ for $i = 0,1,\dots,n-1$. For $j \in K$,  by Corollary~\ref{cor: exists a path connecting points} and since the set $\partial E^-_j$ is, by definition, path connected with $E^\circ_j$ (see equation~\ref{eq: non-active boundary}),
    there are admissible paths  $z(t)_{0\le t\le T}$ such that $z(0) = z^{i}_B$ and $z(T) = z^{i+1}_A$, for $i=0,1,\dots,n-1$.

    Let $t- = \lim_{s \uparrow t} s$ and $t+ = \lim_{s \downarrow t} s$.
    By Proposition~\ref{prop: probability to move through skeleton points} and by Assumption~\ref{ass: connectivity} we have that the set 
    \begin{align*}
    O =  (Z(t_1-) \in  A_{1}  &\mid Z(0) = B_{0})\\
    &\cap_{i=2}^n ( 
    (Z(t_{i-1}+)\in B_{i-1} \mid Z(t_{i-1}-) \in A_{i-1}) \cap(Z(t_i-) \in A_{{i}} \mid Z(t_{i-1} +) \in B_{{i-1}})) 
    \end{align*}
    for a sequence $0< t_1 <t_2<\dots<t_N<\infty$ has positive probability, $\PP(O)>0$.
\end{proof}
With the next lemma, we show that a process started in $(x,v)\in A_x\times \cV$, where $A_x \subset \Omega$ is a open set, is confined in the set $A_x$ for any length $\Delta>0$ with positive probability. This lemma will be used to prove Proposition~\ref{prop: accessibility} in the manuscirpt.
\begin{lemma}\label{app: lemma: confinement}\emph{(Confining the process in a open set)} Let $A_x \subset \Omega$ be an open and non-empty set. For any $T>0$ and $X(0) \in A_x\times \cV$ we have that 
\[
\PP(X(t) \in A_x, \, t = [0,T] \mid Z(0) = z)>0.
\]
\end{lemma}
\begin{proof}
    For a fixed initial point $z = (x,v)\in A_x\times \cV$  there is $r$ such that $B_r(x) \subset A_x$. Define a sequence of points $x_0, x_1,\dots$ such that $x_i \in B_r(x)$ and $x_i \ne x_{i+1}$ for $i=0,1,\dots$. Since $B_r(x)$ is connected and open, we can construct an admissible path $z(t)_{t\ge 0}$ that connects each point $(x_i)_{i=0,1,\dots,M}$ and such that $sum_{i=1}^M \|x_i-x_{i-1}\| = T$. Then, by iterating Proposition~\ref{prop: probability to move through skeleton points} to establish that $P(X(t) \in A, \, t \in [0,T])>0$, for any $T>0$.
\end{proof}

\begin{proposition}\label{app: prop: }\emph{(Proposition~\ref{prop: accessibility} in the manuscript)}
Suppose Assumptions~\ref{ass: conditions on the refreshments}-\ref{ass: conditions on the interior}-\ref{ass: sdb at the boundary}-\ref{ass: connectivity} hold. Then, for every $z \in E$, open set and non-empty set $A \subset E$, and for every $\Delta>0$, there exist an integer $n>0$ such that
\[
\PP_{\Delta n} (z, A) >0. 
\]
\end{proposition}
\begin{proof} 
Fix an initial position $z \in E$ and a open set and non-empty set $A$. By Proposition~\ref{app:  prop reachability 2}, we know that there exists a $t$ with $\PP_t(z, A) > 0$. By Lemma~\ref{app: lemma: confinement}, we know that $\PP_{t + \Delta_0}(z, A) > 0$, for any $\Delta_0> 0$ as we can confine the process in any open ball of its position along its trajectory for any time interval $\Delta_0 \ge 0$. This implies that, for every $\Delta>0$, there exists an integer $n$, such that $\PP_{n \Delta}(z, A) > 0$.
\end{proof}

\subsection{Minorization condition of PDMP samplers}\label{app: minorization}

For the Bouncy particle sampler, the minorization condition was already proven in \textcite{2015arXiv151002451B} in the $d$-dimensional Euclidean space without boundary. Below, we show that the same strategy of \textcite{2015arXiv151002451B} can be easily extended for our setting for the Bouncy particle sampler.
\begin{proposition}
    \label{app: prop: minorization condition}
    \emph{(Proposition~\ref{prop: minorization condition} in the main manuscript)} Suppose Assumptions~\ref{ass: conditions on the refreshments}-\ref{ass: conditions on the interior} hold. There exists $r, \epsilon >0$ such that, for all $\epsilon<t< r$,  there exists a non-empty and open set $C \subset E$ and $\delta>0$ such that for all $z \in C$, 
    \[
    \PP_t(z, A) \ge \delta \Vol(A\cap C), \qquad \forall A \in \cB(E).
    \]
\end{proposition}
For the Bouncy Particle Sampler, without loss of generality, we let $0 \in \Omega$. Then, we choose  $r = \max(r'>0 \colon  B_{r'}(0) \in \Omega)$ (this always exists since $\Omega$ is a open set) and $r_0 < r$. With this choice,  \textcite[proof of Lemma 3 part 1]{2015arXiv151002451B} follows with $\epsilon = r_0/6$ and $C = \cB_{r_0}(0)\times\cB_1(0)$. This is because in \textcite[proof of Lemma 3 part 1]{2015arXiv151002451B}, $\|v(t')\|\le 1, \, t' \in [0,t]$, hence the process cannot hit the boundary before time $r$.
\section{SIR with notifications}\label{app: example in epidemiology}
\subsection{Derivation of the measure in Section~\ref{sec: epidemiology}}\label{app: derivation of the measure in epidemiology}

We now derive the terms $\rho_i(x) \mu_i(\dd x_i)$ (equation~\eqref{eq: target measure not notified} and equation~\eqref{eq: target measure notified}), which can be heuristically interpreted as the distribution of the infection times $X_i$ relative to those individuals $i \in \{1,2,\dots, d\}$ which have not been notified up to time $T$  and those which have been notified before time $T$. To ease the notation, we drop the index $i$ and consider random times $X \in \RR^+$, $\tau^{\star} \in \RR^+$, where $\tau^{\star} = x + \sigma$, with $\sigma \ge 0$ a random variable independent of $x$. Suppose $X$ has density $g(t) = \beta(t) \exp(-B(t) )$, where $B(t) = \int_0^t \beta(s) \, d s$, and $\sigma$ has cdf $F$ and pdf $f$, both with support on $[0,\infty)$.

We wish to determine, for $  t\geq0,\,  T \geq 0$,
\[ \P( X \wedge T \le t \mid \tau^{\star} \ge T).\]
Clearly if $t \ge T$ then this conditional probability is equal to one. It remains to compute, for $t < T$
\[ \P( X \le t \mid \tau^{\star} \ge T).\]

We compute
\begin{align*}
    \P(X\le t \mid \tau^{\star} \ge T) & = \frac{\P(X < t, \tau^{\star} \ge T)}{\P(\tau^{\star} \ge T)} \\
    & = \frac{\P(X < t, X + \sigma \ge T)}{\P(X + \sigma \ge T)} \\
    & = \frac{\int_0^t g(r)  \P(\sigma > T-r) \, d r}{\int_0^T g(r) \P(\sigma > T-r) \, dr + \int_T^{\infty} g(r) \, d r} \\
    & = \frac{\int_0^t g(r) (1 - F(T-r)) \, d r}{\int_0^T g(r)(1-F(T-r)) \, d r + \int_T^{\infty} g(r) \, dr} \\
    & = \frac{\int_0^t \beta(r) \exp(-B(r)) (1-F(T-r)) \, d r}{ \int_0^T \beta(r) \exp(-B(r)) (1-F(T-r)) \, d r + \exp(-B(T))}.
\end{align*}
We see that the random variable $x \wedge T$, conditional on $\tau^{\star} \ge T$, has a Lebesgue density
\[ h(t) = \frac{\beta(t) \exp(-B(t)) (1 - F(T-t))}{ \int_0^T \beta(r) \exp(-B(r)) (1-F(T-r)) \, d r + \exp(-B(T))} \]
on $[0,T]$, and an atomic component of mass
\[ 1 - \int_0^T h(t) \, d t  = \frac{\exp(-B(T))}{\int_0^T \beta(r) \exp(-B(r)) (1-F(T-r)) \, d r + \exp(-B(T))} \]
at $T$. This measure is equal to the measure in  equation~\eqref{eq: target measure not notified}.

The measure in equation~\eqref{eq: target measure notified} of the infection times for notified individuals is easy to derive as
\[
P(X < t \, | \, \tau^\star = c) =  \begin{cases} 0 & \quad t< 0 ,\\
C \int_0^t g(s) f(c - s) \dd s & \quad  0 \le t < c, \\
1 & \quad c \le t, \end{cases}
\]
for some constant $C$, giving the desired result.

\subsection{Computing reflections times}\label{app: Computing reflections times and probabilities to cross discontinuities}

The target measure in equation~\ref{eq: target sinr} can be rewritten as 
\[
\mu(\dd x) \propto L(x) \mu_0(\dd x)
\]
where
\begin{equation*}
  L(x) \propto \exp\left (-\sum_{i= 1}^N B_i(x) \right) \left(\prod_{i\in\cN_T} f(\tau_i^\star - x_i)\beta_i(x)\right)\left(\prod_{i\in\cN_T^c} (1-F(T - x_i))\beta_i(x)\right)
\end{equation*}
and
\begin{equation*}
    \label{eq: reference measure}
     \mu_0(\dd x) = \left(\prod_{i \in \cN_T}\left( \ind_{(0\le x_i\le T)} \dd x_i + \kappa_i\delta_T(\dd x_i)\right) \right) \left(\prod_{i \in \cN^c_T}\ind_{(0\le x_i\le \tau_i^\star)} \dd x_i \right).
\end{equation*}

In the experiments in Section~\ref{sec: epidemiology} we assumed that $f$ and $F$ are respectively the density and the distribution of an exponential r.v. with parameter $\beta$. Then $\partial_{x_i}\log(f(\tau^\star_i - x_i)) = \partial_{x_i}\log(1 - F(T - x_i)) = \beta$. 

The first event time of the $i$th clock of the Zig-Zag process is a Poisson clock with rate $\beta_{i, \mathrm{b}} = (v_i \partial_{x_i}(-\log L(x)))^+$ where
\begin{align*}
\partial_{x_k} (-\log L(x)) &= \sum_j \partial_{\tau_k} B_j(x)  - \sum_{i \in \cN_T} \partial_{x_k} \log f(\tau^\star_i - x_i) - \sum_{i \in \cN_T^c} \partial_{x_k} \log ( 1-  F(T - x_i)) \\
&= \sum_j \partial_{x_k} B_j(x) - \beta
\end{align*}
where
\[
B_j(x) = \sum_{i \ne j} \int_0^T \beta_{i,j}(x[j; t]) \dd t =  \sum_{i \ne j} C_{i,j}((\tau_i^\star \wedge x_j - x_i \wedge x_j ) + \gamma(\tau_i^\circ \wedge x_j - \tau_i^\star \wedge x_j).
\]
Hence the partial derivative of the negative log-likleihood is
\[
\partial_{x_k} (-\log L (x)) = 
\sum_i \partial_{x_k}  B_i(x) - \beta = \sum_{i \ne k} \partial_{x_k} \left(B_{k,i}(x) + B_{i,k}(x)\right) - \beta = \sum_{i \ne k} \cG_{i,k}(x) - \beta
\]
with
\[
\cG_{i,k}(x) = \begin{cases}
    - C_{k,i} & \tau_k < x_i,\\
    C_{i,k}  & x_i < x_k < \tau_i^\star,\\
    \gamma C_{i,k} & \tau_i^\star < x_k < \tau_i^\circ, \\
    0 & \tau_i^\circ < x_k. 
\end{cases}
\]

 Conditioned on the process not hitting a boundary (a discontinuity), the rates $t \mapsto \lambda_{i,\rm{b}}(\phi(t, z)), \, i=1,2,\dots,d$ (equation~\ref{eq: rates zz}) of the Zig-Zag sampler are constant. Hence, if the process is at $z \in E$, we can efficiently simulate the first reflection time simply as $\tau = \min(\tau_1,\tau_2,\dots,\tau_d)$, where  $\tau_i \sim \text{IPP}(t \mapsto \lambda_{i, \rm{b}}(\phi(z, t))$, for $i = 1,2,\dots,d$.

\section[Teleportation rules]{Teleportation rules for hard-sphere models}\label{app: hard-sphere model}
Recall that we aim to sample from the measure 
\[
\mu(\dd x) \propto \exp(-\sum_{i=1}^N\Psi_0(x^{(i)})) \ind_{A}(x) \dd x, \quad \Psi \in \cC^1(\RR^d)
\]
 with  $A = \bigcap_{i =1}^N \bigcap_{\substack{j = 1,2,\dots,N,\\j \ne i}} A_{i,j}$ and 
    \[
    A_{i,j} = \{ x \in \RR^{dN} \colon \|x^{(i)} - x^{(j)}\| \ge (r_i +r_j)\}.
\]
We define the boundary for the process as
\[\partial E^+ = \{ (x,v) \in \partial \Omega\times \RR^{dN} \colon  \langle n(x), v\rangle > 0 \},\] 
where $\partial \Omega  = \bigcup_{i =1}^N \bigcup_{\substack{j = 1\\j \ne i}}^N \partial \Omega_{i,j}$ and
\[
\partial \Omega_{i,j} = \{ x \in \RR^{dN} \colon |x^{(i)} - x^{(j)}| = (r_i +r_j)\}.
\]
We consider only deterministic teleportation rules of the form $\cT(x, \cdot) = \delta_{\kappa(x)}(\cdot)$ for $(x,v) \in \partial E^+$ and informally discuss 3 different choices for $\kappa$, each one being efficient in different scenarios. Here for efficient teleportation rules we mean those that accept the new proposed points with high probability. To that end, we want the teleportation rule 
\begin{itemize}
    \item to swap the centers of two colliding hard-spheres and, after that, to move their center as little as possible to have high acceptance probability, 
    \item to move the total volume of hard-spheres as little as possible in order to reduce the probability that some hard-spheres overlap after teleportation.
\end{itemize}
In the following, we let $(x,v) \in \partial E^+$ with $|x^{(i)} - x^{(j)}| = r_i + r_j$ for some  $i,j$.
\subsection{Swap the centers}\label{app: Swap the centers}
Set, for $\ell = 1,2,\dots, N$
\[
[\kappa(x)]^{(\ell)} = \begin{cases}
    x^{(j)} &\text{ if } \ell = i\\
    x^{(i)} &\text{ if } \ell = j\\
    x^{(\ell)} &\text{ otherwise.}\\
    \end{cases} 
\]
This teleportation swaps $x^{(i)}$ with $x^{(j)}$ (see Figure~\ref{fig: some teleportations}, top panels) and has the advantage that if $\kappa(x) \in A$, then the new point is accepted with probability 1 as $\alpha(x,\kappa(x)) = 1$. 

When the size of two hard-spheres substantially differs, as it is the case when  $r_i \approx 0$ and $r_j$ is large, the large hard-sphere has to be moved as much as the small hard-sphere so that the total volume of hard-spheres which is moved when applying $\kappa$ is large. In the scenario of many hard-spheres surrounding these two hard-spheres, this teleportation rule leads to  invalid proposed configurations with $\kappa(y) \notin A$ with high probability. 

We describe another teleportation rule which is effective when $r_i \approx 0$ and $r_j$ is large.
\subsection{Move only the smaller hard-sphere}\label{app: Move only the smaller ball}
When $r_i\ll r_j$, we propose to move only the smaller hard-sphere by setting, for $\ell = 1,2,\dots, N$
\[
[\kappa(x)]^{(\ell)} = \begin{cases}
    2x^{(j)} - x^{(i)} &\text{ if } \ell = i  \\
    x^{(\ell)} &\text{ otherwise.}
\end{cases} 
\]
see Figure~\ref{fig: some teleportations}, middle panels for an illustration.
The aim here is to not move the larger hard-sphere and therefore maximize the probability that $\kappa(x) \in A$ when the two hard-spheres are surrounded by other hard-spheres. This comes at the cost of having possibly low probability to accept the new points (since $x \mapsto \Psi(x)$ is not constant).

Finally we introduce a third teleportation rule, which is the one used in Section~\ref{sec: teleportation} and can bee seen as a compromise between the teleportation rules described in Section~\ref{app: Swap the centers} and Section~\ref{app: Move only the smaller ball}.
\subsection{Move the smaller hard-sphere more than the larger one}
For $\ell = 1,2,\dots, N$, we set 
\[
[\kappa(x)]^{(\ell)} = \begin{cases}
     x^{(j)} + \frac{x^{(i)} - x^{(j)}}{r_i + r_j}(r_i - r_j) &\text{ if } \ell = i\\
    x^{(i)} + \frac{x^{(j)} - x^{(i)}}{r_i+ r_j}(r_j - r_i)  &\text{ if } \ell = j \\
    x^{(\ell)} &\text{ otherwise,}\\
\end{cases} 
\]
for $\alpha \in [0,1]$.
This teleportation rule moves the smaller hard-sphere more and the larger hard-sphere in order to increase the probability that $\kappa(x) \in A$ while  trying to move as little as possible the center of the hard-spheres, see Figure~\ref{fig: some teleportations}, bottom panel. Notice that, if $r_i = r_j$, this teleportation coincides with the teleportation in Section~\ref{app: Swap the centers}.
\begin{figure}
    \centering
    \includegraphics[width = 0.9\textwidth]{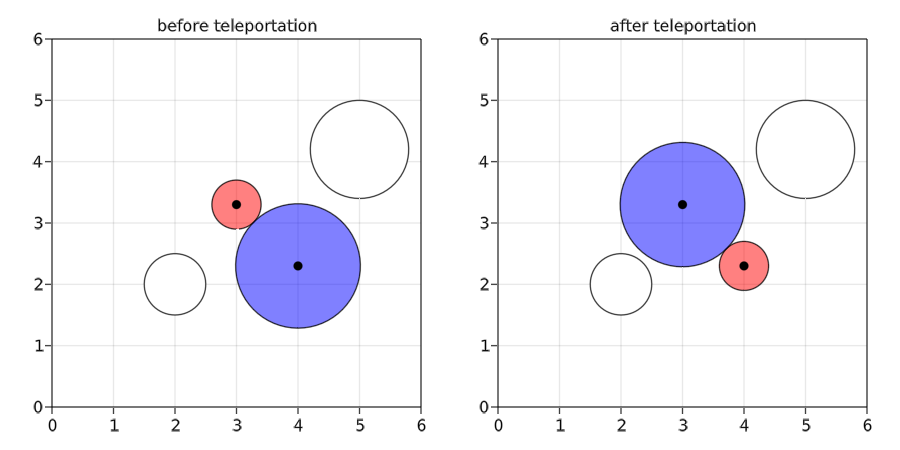}
    \includegraphics[width = 0.9\textwidth]{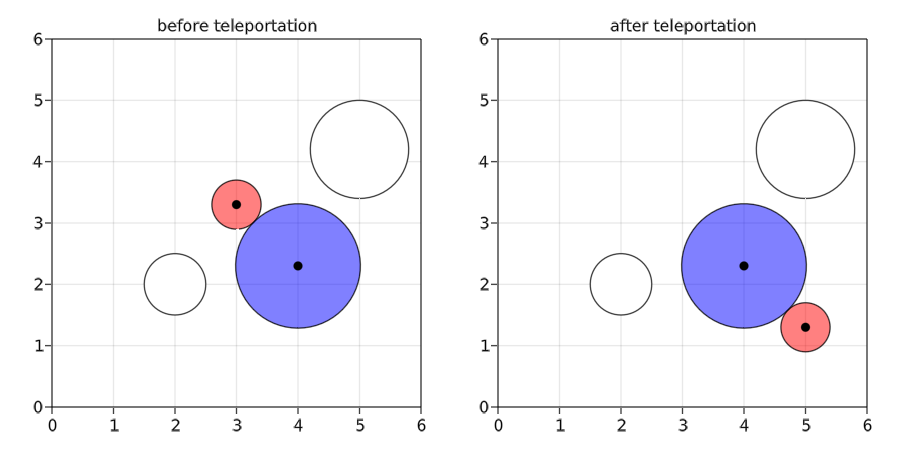}
    \includegraphics[width = 0.9\textwidth]{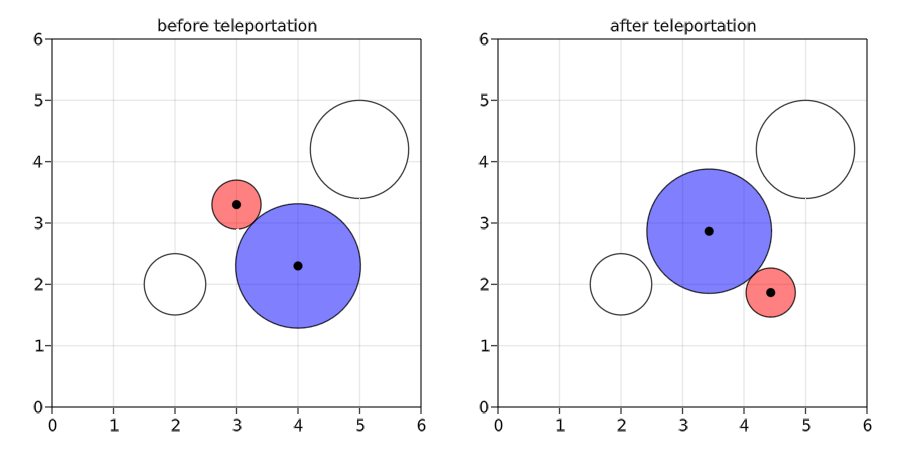}
    \caption{Illustration of 3 different teleportation rules (from Top to Bottom panel). Left panels: configuration before teleportation; right panels: configuration after teleportation. Top panels: teleportation by swapping the centers of 2 hard-spheres. Middle panels: teleportation by moving only the smaller hard-sphere. Bottom-panel: teleportation which moves more the smaller hard-sphere than the bigger one.}
    \label{fig: some teleportations}
\end{figure}

\end{document}